\documentclass[acmsmall]{acmart}

\acmJournal{TOSEM}
\acmVolume{1}
\acmNumber{} 
\acmArticle{1}
\acmYear{}
\acmMonth{1}
\acmDOI{} 
\startPage{1}

\setcopyright{none}

\usepackage{mathtools}
\usepackage{booktabs}   
\usepackage{subcaption} 
\newsavebox{\tablebox}

\def\dprog{\mathit{DProg}}
\def\dspec{\mathit{Spec^e}}
\def\nprog{\mathit{NProg}}
\def\nspec{\mathit{Spec^s}}

\def\spec{\mathit{Spec}}

\def\dtsat{\models_{\tot}}
\def\dpsat{\models_{\pal}}
\def\ntsat{\models_{\tot}}
\def\npsat{\models_{\pal}}
\def\atsat{\models_{\tot}}
\def\apsat{\models_{\pal}}

\def\dtrefine{\le_T^e}
\def\dprefine{\le_P^e}
\def\ntrefine{\le_T^s}
\def\nprefine{\le_P^s}
\def\atrefine{\le_T^p}
\def\aprefine{\le_P^p}
\usepackage{multirow}
\usepackage{makecell}
\usepackage{adjustbox}
\usepackage{xcolor} 
\begin{document}
	
	\title{Refinement orders for quantum programs}
	
	\author{Yuan Feng}
	\affiliation{
		\institution{Department of Computer Science and Technology, Tsinghua University}
		\city{Beijing}
		\country{China}
	}
	\email{yuan_feng@tsinghua.edu.cn}
	\author{Li Zhou}
	\affiliation{
		\institution{State Key Laboratory of Computer Science, Institute of Software, Chinese Academy of Sciences}
		\city{Beijing}
		\country{China}
	}
	\email{zhouli@ios.ac.cn}
	
	\begin{abstract}

Refinement is a fundamental technique in the verification and systematic development of computer programs. It supports a disciplined approach to software construction through stepwise refinement, whereby an abstract specification is gradually transformed into a concrete implementation that satisfies the desired requirements. Central to this methodology is the notion of a refinement order, which guarantees that each refinement step preserves program correctness.

This paper presents the first comprehensive study of refinement orders for quantum programs, covering both deterministic and nondeterministic settings under total and partial correctness criteria. We investigate three natural classes of quantum predicates: projectors, representing qualitative properties; effects, capturing quantitative properties; and sets of effects, modeling demonic nondeterminism. For deterministic quantum programs, we show that refinement with respect to effect-based and set-of-effects based specifications coincides with the standard complete-positivity order on superoperators, whereas refinement induced by projector-based specifications can be characterized by the linear span of Kraus operators. For nondeterministic quantum programs with set-of-effects based specifications, we establish precise correspondences with classical domain-theoretic notions: the Smyth order characterizes refinement under total correctness, while the Hoare order characterizes refinement under partial correctness. Moreover, effect-based and projector-based specifications lead to strictly weaker refinement orders. 

From a theoretical perspective, our results provide a solid semantic foundation for the development of quantum refinement calculi. From a practical standpoint, they offer concrete guidance for quantum program designers in selecting appropriate predicate classes and correctness notions to support their intended refinement goals.

	\end{abstract}

%


	\maketitle
	
	\newcommand \sker[1] {\mathcal{N}\left(#1\right)}
\newcommand \assert[1] {\mathbf{assert}\ #1}
\newcommand \alert[1] {{\color{red} #1}}
\newcommand {\lfp} {\mathbf{lfp}\ }
\newcommand {\gfp} {\mathbf{gfp}\ }
\newcommand {\E} {\supoprset}
\newcommand {\F} {\mathbb{F}}
\newcommand {\tot} {\mathit{tot}}
\newcommand {\pal} {\mathit{par}}
\newcommand{\conf}[2]{\left\<#1, #2\right\>}
\newcommand{\cqs}[2]{\left\<#1, #2\right\>}
\newcommand{\spdomain}[1]{\p(#1)}
\newcommand{\upcl}{\,\uparrow\!}
\newcommand{\downcl}{\,\downarrow\!}
\newcommand{\cccl}{cc.}
\newcommand {\op} {\mathit{op}}
\newcommand{\vecspace}{V}
\newcommand{\refine}{\le}

\newcommand{\compstate}{\mathrm{Val}}
\newcommand{\compts}{\mathrm{Comp}}

\newcommand{\pcom}[1]{\ {}_{#1}\!\!\oplus}

\newcommand {\nondet} {\ \square\ }

\newcommand {\empstr} {\Lambda}

\newcommand {\qcf}[1] {{\sf{#1}}}

\newcommand {\qc}[1] {{\sf{#1}}}
\def\>{\ensuremath{\rangle}}
\def\<{\ensuremath{\langle}}
\def\sl {\ensuremath{\llparenthesis}}
\def\sr{\ensuremath{\rrparenthesis}}
\def\-{\ensuremath{\textrm{-}}}
\def\ott{t}
\def\otu{u}
\def\ots{s}
\def\apply{\mathrel{*\!\!=}}
\def\dhall{\d(\h_{\QVar})}

\def\comm{\ensuremath{\leftrightarrow^*}}
\def\reach{\ensuremath{\rightarrow^*}}

\def\ctp{P}
\def\ctq{Q}

\def\change{\ensuremath{\mathit{change}}}

\def\qVar{\ensuremath{\mathit{qv}}}
\def\qv{\ensuremath{\mathit{qv}}}
\def\cVar{\ensuremath{\mathit{cv}}}
\def\QVar{\ensuremath{\mathcal{V}}}
\def\CVar{\ensuremath{\mathit{cVar}}}
\def\Var{\ensuremath{\mathit{var}}}
\def\Chan{\mathit{chan}}
\def\cChan{\mathit{cChan}}
\def\qChan{\mathit{qChan}}
\def\BExp{\mathit{BExp}}
\def\Exp{\mathit{Exp}}

\def\fdmu{\Delta}
\def\fdnu{\dnu}
\def\fdomega{\domega}

\def\dmu{\mu}
\def\dnu{\nu}
\def\domega{\omega}
\def\expect{\mathbb{E}}
\def\preexpect{\mathrm{pre}\mathbb{E}}

\def\rassign{:=_{\$}}
\def\fpi{\widehat{\pi}}
\def\h{\ensuremath{\mathcal{H}}}
\def\p{\ensuremath{\mathcal{P}}}
\def\l{\ensuremath{\mathcal{L}}}
\def\g{\ensuremath{\mathcal{G}}}
\def\lh{\ensuremath{\mathcal{L(H)}}}
\def\dh{\ensuremath{\mathcal{D(H})}}
\def\dhv{\ensuremath{\d(\h_{\QVar})}}
\def\q{\bold Q}
\def\Q{\ensuremath{\mathbb Q}}
\def\P{\ensuremath{\mathbb P}}
\def\SO{\ensuremath{\mathcal{SO}}}
\def\HP{\ensuremath{\mathcal{HP}}}
\def\hpe{\ensuremath{\mathcal{\e}}}

\def\r{\ensuremath{\mathcal{R}}}
\def\R{\ensuremath{\mathbb{R}}}
\def\m{\ensuremath{\mathcal{M}}}
\def\u{\ensuremath{\mathcal{U}}}
\def\k{\ensuremath{\mathcal{K}}}
\def\K{\ensuremath{\mathfrak{K}}}
\def\S{\ensuremath{\mathfrak{S}}}
\def\s{\ensuremath{\mathcal{S}}}
\def\t{\ensuremath{\mathcal{T}}}
\def\u{\ensuremath{\mathcal{U}}}
\def\U{\ensuremath{\mathfrak{U}}}
\def\L{\ensuremath{\mathfrak{L}}}
\def\x{\ensuremath{\mathcal{X}}}
\def\y{\ensuremath{\mathcal{Y}}}
\def\z{\ensuremath{\mathcal{Z}}}
\def\v{\ensuremath{\mathcal{V}}}

\def\st{\ensuremath{\mathfrak{t}}}
\def\su{\ensuremath{\mathfrak{u}}}
\def\ss{\ensuremath{\mathfrak{s}}}

\def\ra{\ensuremath{\rightarrow}}
\def\a{\ensuremath{\mathcal{A}}}
\def\b{\ensuremath{\mathcal{B}}}
\def\c{\ensuremath{\mathcal{C}}}

\def\e{\ensuremath{\mathcal{E}}}
\def\f{\ensuremath{\mathcal{F}}}
\def\l{\ensuremath{\mathcal{L}}}
\def\X{\mbox{\bf{X}}}
\def\N{\mathbb{N}}
\def\sreal{\mathbb{R}}
\def\Z{\mathbb{Z}}

\def\qzz{\ensuremath{|0\>_q\<0|}}
\def\qoo{\ensuremath{|1\>_q\<1|}}
\def\qzo{\ensuremath{|0\>_q\<1|}}
\def\qoz{\ensuremath{|1\>_q\<0|}}
\def\qii{\ensuremath{|i\>_q\<i|}}
\def\qiz{\ensuremath{|i\>_q\<0|}}
\def\qzi{\ensuremath{|0\>_q\<i|}}

\def\quzz{\ensuremath{|0\>_{\bar{q}}\<0|}}
\def\quoo{\ensuremath{|1\>_{\bar{q}}\<1|}}
\def\quzo{\ensuremath{|0\>_{\bar{q}}\<1|}}
\def\quoz{\ensuremath{|1\>_{\bar{q}}\<0|}}
\def\quii{\ensuremath{|i\>_{\bar{q}}\<i|}}
\def\quiz{\ensuremath{|i\>_{\bar{q}}\<0|}}
\def\quzi{\ensuremath{|0\>_{\bar{q}}\<i|}}

\DeclarePairedDelimiter{\ceil}{\lceil}{\rceil}

\def\d{\ensuremath{\mathcal{D}}}
\def\dh{\ensuremath{\mathcal{D(H)}}}
\def\lh{\ensuremath{\mathcal{L(H)}}}
\def\le{\ensuremath{\sqsubseteq}}
\def\ge{\ensuremath{\sqsupseteq}}
\def\eval{\ensuremath{{\psi}}}
\def\aeq{\ensuremath{{\ \equiv\ }}}
\def\osnt{\ensuremath{\sl \ott, \e\sr}}
\def\snt{\st}
\def\snti{\ensuremath{\sl \ott_i, \e_i\sr}}
\def\osnu{\ensuremath{\sl \otu, \f\sr}}
\def\osns{\ensuremath{\sl s, \g\sr}}
\def\snu{\su}
\def\fdist{\ensuremath{\d ist_\h}}
\def\dist{\ensuremath{Dist}}
\def\wtx{\ensuremath{\widetilde{X}}}

\def\bv{1{v}}
\def\bV{\mathbf{V}}
\def\bf{\mathbf{f}}
\def\bw{\mathbf{w}}
\def\zo{\mathbf{0}}
\def\bX{\mathbf{X}}
\def\bDelta{\mathbf{\Delta}}
\def\bdelta{\boldsymbol{\delta}}
\def\next{\mathcal{X}}
\def\until{\mathcal{U}}

\def\leqI{\ensuremath{\mathcal{SI}(\h)}}
\def\leqIq{\ensuremath{\mathcal{SI}_{\eqsim}(\h)}}
\def\oact{\ensuremath{\alpha}}
\def\oactb{\ensuremath{\beta}}
\def\sact{\ensuremath{\gamma}}
\def\fpi{\ensuremath{\widehat{\pi}}}
\newcommand{\supp}[1]{\ensuremath{\left\lceil{#1}\right\rceil}}
\newcommand{\support}[1]{\lceil{#1}\rceil}

\newcommand{\abis}{\stackrel{\lambda}\approx}
\newcommand{\abisa}[1]{\stackrel{#1}\approx}
\newcommand {\qbit} {\mbox{\bf{new}}}

\renewcommand{\theenumi}{(\arabic{enumi})}
\renewcommand{\labelenumi}{\theenumi}
\newcommand{\tr}{{\rm tr}}
\newcommand{\rto}[1]{\stackrel{#1}\rightarrow}
\newcommand{\orto}[1]{\stackrel{#1}\longrightarrow}
\newcommand{\srto}[1]{\stackrel{#1}\longmapsto}
\newcommand{\sRto}[1]{\stackrel{#1}\Longmapsto}

\newcommand{\ass}[3]{\left\{#1\right\}#2\left\{#3\right\}}
\newcommand{\andor}{\ \&\ }

\newcommand {\true} {\ensuremath{{\mathbf{true}}}}
\newcommand {\false} {\ensuremath{{\mathbf{false}}}}
\newcommand {\abort}{\ensuremath{{\mathbf{abort}}}}
\newcommand {\sskip} {\mathbf{skip}}

\newcommand {\then} {\ensuremath{\mathbf{then}}}
\newcommand {\eelse} {\ensuremath{\mathbf{else}}}
\newcommand {\while} {\ensuremath{\mathbf{while}}}
\newcommand {\ddo} {\ensuremath{\mathbf{do}}}
\newcommand {\pend} {\ensuremath{\mathbf{end}}}
\newcommand {\inv} {\ensuremath{\mathbf{inv}}}

\newcommand {\mymeas} {\mathbf{measure}}

\newcommand {\fail} {\mathbf{fail}}
\newcommand {\iif} {\mathbf{if}}
\newcommand {\fii} {\mathbf{fi}}
\newcommand {\od} {\mathbf{od}}
\def\mstm{\iif\ b\ \then\ S_1\ \eelse\ S_0\ \pend}
\def\wstm{\while\ b\ \ddo\ S\ \pend}

\newcommand\pmeasstm[3]{\iif\ #1\ \then\ #2\ \eelse\ #3\ \pend}
\def\pmstm{\iif\ P[\bar{q}]\ \then\ S_1\ \eelse\ S_0\ \pend}
\def\pwstm{\while\ P[\bar{q}]\ \ddo\ S\ \pend}

\newcommand\measstm[3]{\iif\ #1\ra\ #2\ \square\ \neg (#1)\ra #3\ \fii}

\newcommand\whilestm[1]{\while\ P[\bar{q}]\ \ddo\ #1\ \pend}

\newcommand\alterex{\iif\ B_1\ra S_1 \square\ldots\square B_n\ra S_n\ \fii}

\newcommand\altercom{\iif\ \square_{i=1}^n B_i\ra S_i\ \fii}

\newcommand\repex{\ddo\ B_1\ra S_1 \square\ldots\square B_n\ra S_n\ \od}

\newcommand\repcom{\ddo\ \square_{i=1}^n B_i\ra S_i\ \od}

\newcommand\seqcom{\ddo\ \square_{i=1}^n B_i;\alpha_i\ra S_i\ \od}

\newcommand {\spann} {\mathrm{span}}

\newcommand{\rrto}[1]{\xhookrightarrow{#1}}
\newcommand{\con}[3]{\iif\ {#1}\ \then\ {#2}\ \eelse\ {#3}\ \pend}

\newcommand{\Rto}[1]{\stackrel{#1}\Longrightarrow}
\newcommand{\nrto}[1]{\stackrel{#1}\nrightarrow}

\newcommand{\Rhto}[1]{\stackrel{\widehat{#1}}\Longrightarrow}
\newcommand{\define}{\ensuremath{\triangleq}}
\newcommand{\rsim}{\simeq}
\newcommand{\obis}{\approx_o}
\newcommand{\sbis}{\ \dot\approx\ } 
\newcommand{\stbis}{\ \dot\sim\ } 
\newcommand{\nssbis}{\ \dot\nsim\ } 

\newcommand{\bis}{\sim}
\newcommand{\rat}{\rightarrowtail}
\newcommand{\wbis}{\approx}
\newcommand{\id}{\mathcal{I}}
\newcommand{\stet}[1]{\{ {#1}  \}  } 
\newcommand{\unw}[1]{\stackrel{{#1}}\sim}
\newcommand{\rma}[1]{\stackrel{{#1}}\approx}

\def\step{\textsf{step}}
\def\obs{\textsf{obs}}
\def\dom{\textsf{dom}}
\def\purge{\textsf{ipurge}}
\def\source{\textsf{sources}}
\def\cnt{\textsf{cnt}}
\def\read{\textsf{read}}
\def\alter{\textsf{alter}}
\def\dirac#1{\delta_{#1}}

\def\tybool{\ensuremath{\mathbf{Boolean}}}
\def\tyint{\ensuremath{\mathbf{Integer}}}
\def\tyqubit{\ensuremath{\mathbf{Qubit}}}
\def\tyqudit{\ensuremath{\mathbf{Qudit}}}
\def\tyqureg{\ensuremath{\mathbf{Qureg}}}
\def\tyunitreg{\ensuremath{\mathbf{Unitreg}}}
\def\type{\ensuremath{\mathit{type}}}

\def\qstate{\rho}
\def\qassert{\Theta}
\def\qassertp{\Psi}
\def\casserts{\a}
\def\cstate{S}
\def\cstates{\prog}
\def\cassert{p}
\def\emptydis{\bot}
\def\qset{Q}
\def\qsetp{R}
\def\Exp{\mathrm{Exp}}

\def\supoprset{\mathbb{E}}

\def\leinf{\preceq_{\mathsf{dem}}}
\def\lesup{\preceq_{\mathsf{ang}}}

\def\geinf{\ge_{\mathit{inf}}}
\def\qstates{\s_V}
\def\qasserts{\a_V}
\def\qstatesh#1{\d(\mathcal{H}_{#1})}

\def\qassertsh#1{\mathcal{A}_{#1}}

\def\qstatesp{\mathcal{S}(\h')}

\newcommand\prog{\mathit{Prog}}
\def\ph{\ensuremath{\mathcal{P}(\h)}}
\def\phv{\ensuremath{\mathcal{P}(\h_V)}}

\def\l{\mathcal{L}}
\def\k{\mathcal{K}}
\def\qmc {\color{red}}
\def\dtmc {\color{black}}
\newcommand{\ysim}[1]{\stackrel{#1}\sim}
\def\z{\mathbf{0}}
\newcommand{\TRANDA}[3]{#1\xrightarrow{#2}_{{\sf D}}#3}
\def\pdist{\mathit{pDist}}

\def\C{\mathbb{C}}

\newcommand{\subs}[2]{{#2}/{#1}}

\def \Rm#1{\mbox{\rm #1}}
\def \lsem      {\raise1pt\hbox{\Rm {[\kern-.12em[}}}
\def \rsem      {\raise1pt\hbox{\Rm {]\kern-.12em]}}}
\def \sem#1{\mbox{\lsem$#1$\rsem}}

\newtheorem{remark}{Remark}

	\section{Introduction}
	
	The increasing complexity and criticality of modern software systems demand rigorous development methodologies to ensure correctness and reliability. Refinement calculus~\cite{dijkstra1976discipline,morgan1994programming,back1998refinement,back1981correct} provides a structured and systematic approach to program construction by enabling stepwise refinement from abstract specifications to concrete implementations while preserving correctness at each stage. This methodology not only enhances software reliability but also promotes important software engineering principles such as modularity, maintainability, and code reuse~\cite{wirth1971program,kourie2012correctness,mciver2020correctness}. Beyond program development, refinement also finds applications in program verification~\cite{dijkstra1968constructive,correll1978proving} and compilation~\cite{hoare1992refinement,fidge1997modelling}.
	
	At the heart of refinement calculus lies the concept of a \emph{refinement order}, which formalizes the idea that one program refines another if it meets all the specifications that the latter satisfies~\cite{back1981correct}. Formally, let \( \prog \) denote the set of programs and \( \spec \) the set of specifications. For any two programs \( S, S' \in \prog \), we say that \( S \) refines \( S' \) (or that \( S' \) is refined by \( S \)) if, for every specification \( X \in \spec \),
	\begin{equation}\label{eq:metadef}
		S' \models X \quad \implies \quad S \models X
	\end{equation}
	where \( \models \) denotes the satisfaction relation. Stepwise development from a specification \( X \) then involves constructing a sequence of programs \( S_1, \ldots, S_n \) such that \( S_1 \models X \) and \( S_{i+1} \) refines \( S_i \) for \( 1 \leq i \leq n-1 \). The defining equation~\eqref{eq:metadef} ensures that the final program \( S_n \), typically a fully executable implementation, necessarily satisfies \( X \).
	
	\subsection*{Need for Refinement of Quantum Programs}
	
	Quantum computing represents a paradigm shift in computation, leveraging the principles of quantum mechanics, such as superposition, entanglement, and the inherently probabilistic and destructive nature of quantum measurement, to solve problems that are infeasible for classical computers~\cite{nielsen2002quantum,shor1994algorithms,grover1996fast,harrow2009quantum,childs2003exponential}. As quantum technologies mature and quantum computers with hundreds of qubits become available, the development of quantum algorithms and applications is transitioning from theoretical exploration to practical implementation. This transition amplifies the need for rigorous methodologies to design, verify, and optimize quantum programs.
	
	The stakes are particularly high in quantum computing: quantum programs are notoriously difficult to debug due to the no-cloning theorem and the destructive nature of measurement. Moreover, quantum resources, both the coherence time of quantum hardware and the number of available qubits, are precious and limited. These factors make correctness-by-construction approaches, such as refinement calculus, especially attractive for quantum program development. By ensuring correctness at each refinement step, we can avoid costly debugging cycles and make efficient use of scarce quantum resources.
	
	While refinement calculus has been well-established in classical programming, its extension to quantum programs presents unique challenges due to the fundamental differences between quantum and classical systems. 
	A central challenge in developing refinement calculus for quantum programs is the lack of consensus on the appropriate choice of quantum predicates. Unlike classical programs, where predicates are typically expressed as Boolean-valued functions over states, quantum mechanics offers multiple reasonable choices for describing quantum state properties. The existing literature has explored various options, including projectors (Hermitian operators with eigenvalues in $\{0,1\}$)~\cite{zhou2019applied, unruh2019quantum, unruh2019quantumr, feng2023abstract}, effects (Hermitian operators with eigenvalues in $[0,1]$)~\cite{d2006quantum, ying2012floyd,ying2024foundations}, and sets of effects~\cite{feng2023verification}. Each choice leads to different expressive power and different notions of program refinement.
	
	Furthermore, as in classical refinement calculus, we must distinguish between total correctness (which accounts for termination) and partial correctness (which does not). The interplay between these different predicate types and correctness criteria gives rise to a rich variety of refinement orders for quantum programs. However, the relationships between these orders and their practical implications remain poorly understood. 
	
	\subsection*{Our Approach and Contributions}
	
	Inspired by the foundational work of Back~\cite{back1981correct} on classical programs, this paper aims to address these challenges by systematically characterizing refinement orders for quantum programs under various combinations of predicate types and correctness criteria. We adopt a semantic approach that abstracts away from specific language syntax, ensuring that our results are language-independent and broadly applicable. 
	
	To be specific, we model deterministic quantum programs as completely positive and trace-nonincreasing (CPTN) super-operators. In contrast to quantum channels, represented by completely positive and trace-preserving (CPTP) super-operators in quantum information theory, we allow trace-nonincreasing behavior to account for possible non-termination in quantum programs. This choice aligns with standard practice in quantum program semantics and is justified by the fact that the quantum while-language (and its extensions) can implement any CPTN super-operator on finite-dimensional Hilbert spaces~\cite{ying2012floyd, chadha2006reasoning, Kakutani:2009, unruh2019quantum, feng2007proof, rand2019verification,feng2023verification}. Correspondingly, nondeterministic quantum programs are modeled as sets of CPTN super-operators.
	
	Specifications in our framework are represented as pairs of predicates \( (X, Y) \), where \( X \) denotes the precondition and \( Y \) denotes the postcondition. We systematically investigate how the choice of predicate type—projectors, effects, or sets of effects—affects the resulting refinement orders, for both total and partial correctness. Table~\ref{tb:defs} summarizes the various program elements and the corresponding mathematical objects we study in this paper.

	\begin{table}[t]
		\centering
		\renewcommand{\arraystretch}{1.5}
		\resizebox{\textwidth}{!}{%
			\begin{tabular}{c|c|c}
				\hline
				Program Elements & Mathematical Objects & Orders \\
				\hline \hline
				\makecell{Projector predicates $P, Q$, etc.\\ Effect predicates $M, N$, etc.} 
				&\makecell{Linear operators\\ in $\mathcal{L}(\mathcal{H})$} 
				& \makecell{L\"{o}wner order (based on positivity):\\ $A \le B$ iff $\ B - A$ is positive} \\
				\hline
				\makecell{Deterministic quantum\\ programs $\mathcal{E}, \mathcal{F}$, etc.} 
				&	\makecell{Super-operators\\ in $\mathcal{L}(\mathcal{H}) \to \mathcal{L}(\mathcal{H})$} 
				& \makecell{Order based on complete positivity:\\ $\mathcal{E} \le \mathcal{F}$ iff $\ \mathcal{F} - \mathcal{E}$ is completely positive} \\
				\hline
				\makecell{Set-of-effect predicates\\ $\qassert, \qassertp$, etc.} 
				&\makecell{Subsets of linear operators\\ in $2^{\mathcal{L}(\mathcal{H})}$} 
				&  \multirow{2}{15em}{\makecell{Hoare and Smyth orders:\smallskip\\
						$X \leq_H Y$ iff $\ \forall x \in X,\ \exists y \in Y:\ x \le y$\\
						$X \leq_S Y\ $ iff $\ \forall y \in Y,\ \exists x \in X:\ x \le y$}} \\
				\cline{1-2}
				\makecell{Nondeterministic quantum\\ programs $\mathbb{E}, \mathbb{F}$, etc.} 
				&\makecell{Subsets of super-operators\\ in $2^{\mathcal{L}(\mathcal{H}) \to \mathcal{L}(\mathcal{H})}$} 
				& \\
				\hline
		\end{tabular}}
		\caption{Main program elements, the corresponding mathematical objects, and the orders investigated in this paper.}\label{tb:defs}
	\end{table}

	Our key contributions include the following (see Table~\ref{tb:results} for a comprehensive summary):
	
	\begin{enumerate}
		\item For \emph{deterministic} programs, we demonstrate that when effects are taken as state predicates, the refinement orders align precisely with intrinsic orders based on complete positivity. Specifically, for CPTN super-operators $\e$ and $\f$: (i) $\f$ refines $\e$ under total correctness iff $\ \e\le \f$, and (ii) $\f$ refines $\e$ under partial correctness iff $\ \f\le \e$, where $\e\le \f$ means that $\f-\e$ is completely positive. Moreover, we show that strengthening state predicates from effects to sets of effects does not alter the refinement order, whereas weakening them to projectors results in a strictly weaker refinement order. 
		
		\item For \emph{nondeterministic} quantum programs, when sets of effects are taken as state predicates, we establish precise connections between refinement orders and the well-known Hoare and Smyth orders from domain theory. Specifically, for sets of CPTN super-operators $\E$ and $\F$: (i) $\F$ refines $\E$ under total correctness iff $\ \E \leq_S \F$ (Smyth order), and (ii) $\F$ refines $\E$ under partial correctness iff $\ \F \leq_H \E$ (Hoare order). Furthermore, we show that weakening state predicates to effects or projectors leads to strictly weaker refinement orders.

		\item By examining all combinations of predicate types (projectors, effects, sets of effects) and correctness criteria (total, partial) for both deterministic and nondeterministic programs, we provide a complete landscape of refinement orders for quantum programs. 
	\end{enumerate}
	
	\begin{table}[t]
		\centering
		\renewcommand{\arraystretch}{1.5}
		\resizebox{\textwidth}{!}{%
			\begin{tabular}{c|c|c|c|c|c|c}
				\hline
				\multirow{2}{3em}{ } &  \multicolumn{2}{c|}{Projector predicates}  &   \multicolumn{2}{c|}{Effect predicates}  &   \multicolumn{2}{c}{Set-of-effect predicates}\\ \cline{2-7}
				& Total $\le_T^p$  & Partial $\le_P^p$   & Total $\le_T^e$  & Partial $\le_P^e$  &  Total $\le_T^s$  & Partial $\le_P^s$ \\
				\hline\hline
				\makecell{\\Deterministic\\ programs} & \makecell{Weaker\\ than effects\\ $\le_T^e {\subsetneq} \le_T^p $} & \makecell{Weaker\\ than effects\\ $\le_P^e {\subsetneq} \le_P^p $}   & \makecell{Char. by\\ CP-order\\ $\le_T^e {=} \le$} & \makecell{Char. by\\ CP-order\\$\le_P^e {=} \ge$} &  \makecell{Equivalent\\ to effects\\  $\le_T^s {=} \le_T^e$} &  \makecell{Equivalent\\ to effects\\  $\le_P^s {=} \le_P^e$}  
				\\  \cline{2-7}
				&  \multicolumn{2}{c|}{Full char.  (Thm~\ref{thm:detproj})} &  \multicolumn{2}{c|}{Full char.  (Thm~\ref{thm:det})} &  \multicolumn{2}{c}{Cor.~\ref{cor:nondet}}	\\ 
				\hline
				\makecell{\\ Nondeterministic\\ programs} &  \makecell{Weaker\\ than effects\\ $\le_T^e {\subsetneq} \le_T^p $} & \makecell{Weaker\\ than effects\\ $\le_P^e {\subsetneq} \le_P^p $}   & \makecell{Weaker\\ than\\
					effect sets\\ $\le_T^s {\subsetneq} \le_T^e $} &\makecell{Weaker\\ than\\
					effect sets\\ $\le_P^s {\subsetneq} \le_P^e $} &  \makecell{Char. by\\ Smyth order\\  $\le_T^s {=} \leq_S$} &  \makecell{Char. by\\ Hoare order\\  $\le_P^s {=} \geq_H$}   \\  \cline{2-7}
				&  \multicolumn{2}{c|}{Full char. (Thm~\ref{thm:ndetproj_complete})} &  \multicolumn{2}{c|}{Prop~\ref{prop:simpe}} &  \multicolumn{2}{c}{Full char. (Thm.~\ref{thm:nondet})}	\\ 
				\hline
		\end{tabular}}
		\caption{Summary of main results: relationships between refinement orders for different types of quantum programs (rows) under various forms of quantum predicates (columns). Each combination is analyzed with respect to both total and partial correctness.} \label{tb:results}
	\end{table}
	
	These contributions lay the semantic foundations for developing practical refinement calculi for quantum programs and provide theoretical insights that can guide the design of quantum programming languages and verification tools.
	
	\paragraph{Paper Organization}
	The remainder of this paper is organized as follows. After reviewing related work, Section~\ref{sec:preliminaries} provides necessary preliminaries on quantum computing, including key concepts such as partial density operators, quantum measurement, and super-operators. Section~\ref{sec:predicate} introduces various notions of quantum predicates and discusses their use in specifying properties of quantum states. Section~\ref{sec:deterministic} examines refinement orders for deterministic quantum programs, characterizing these orders in terms of intrinsic orders based on complete positivity and linear sum of Kraus operators. We further analyze the impact of different state predicates on the refinement orders. Section~\ref{sec:nondeterministic} extends the analysis to nondeterministic quantum programs, establishing connections with the Hoare and Smyth orders from domain theory. Finally, Section~\ref{sec:conclusion} concludes the paper with a summary of our findings and directions for future research.
	
	\paragraph{Related Work}
	
	The seminal work by Back~\cite{back1981correct} established the theoretical foundations for refinement calculus in classical computing, studying refinement orders for both deterministic and nondeterministic programs with respect to partial and total correctness. Our work extends this framework to the quantum setting. However, the quantum case presents significantly greater complexity due to the variety of possible state predicates. While classical programs typically use Boolean predicates over states, quantum computing admits multiple mathematically valid choices—orthogonal projectors, effects, and sets of effects—each with different expressive power and implications for refinement. This multiplicity of choices makes the landscape of quantum refinement orders considerably more intricate than the classical case.
	
	Our work takes a semantic, language-independent approach to characterizing refinement orders. Complementary to our theoretical investigation, recent studies have developed practical refinement calculi with concrete rules for quantum program development within specific language frameworks. Peduri et al.~\cite{peduri2025qbc} formulated refinement rules for a quantum while-language extended with prescription statements, employing effects as state assertions. Similarly, Feng et al.~\cite{feng2023refinement} developed refinement rules for a similar language but using projectors as state predicates. Our systematic comparison of different predicate choices provides theoretical justification for these practical approaches and clarifies the relationships between them.

	\section{Preliminaries}
	\label{sec:preliminaries}
	
	This section introduces the necessary mathematical foundations from quantum computing and domain theory required for understanding our results. Readers seeking more comprehensive background on quantum computing are referred to~\cite[Chapter 2]{nielsen2002quantum}.
	
	\subsection{Basic Quantum Computing}
	
	\subsubsection{Quantum State}
	In von Neumann's formulation of quantum mechanics~\cite{vN55}, any quantum system with finite degrees of freedom is associated with a finite-dimensional complex Hilbert space $\h$, called its \emph{state space}. When $\dim(\h) = 2$, the quantum system is called a \emph{qubit}, the quantum analogue of a classical bit.
	
	A \emph{pure state} of the system is represented by a unit vector $|\psi\rangle = \sum_{i=0}^{d-1} \alpha_i |i\rangle$ in $\h$, where $\{|i\rangle : 0 \leq i < d\}$ forms an \emph{orthonormal basis}, the coefficients $\alpha_i \in \mathbb{C}$ are complex amplitudes, and the normalization condition $\sum_i |\alpha_i|^2 = 1$ holds.\footnote{We adopt the Dirac notation standard in quantum computing, where vectors in $\h$ are written in ket form as $|\psi\rangle$, $|i\rangle$, etc.} 	
	In practice, quantum systems are often in \emph{mixed states}, which arise from probabilistic ensembles or entanglement with external systems. A mixed state is described by a \emph{density operator} $\rho = \sum_k p_k |\psi_k\rangle\langle\psi_k|$, where each $|\psi_k\rangle$ is a pure state occurring with classical probability $p_k$ (with $\sum_k p_k = 1$). The notation $\langle\psi|$ denotes the \emph{adjoint} (conjugate transpose) of $|\psi\rangle$, and $|\psi\rangle\langle\psi|$ represents the \emph{outer product}: if $|\psi\rangle$ is a column vector, then $\langle\psi|$ is the corresponding row vector, and $|\psi\rangle\langle\psi|$ is their matrix product. Density operators are positive and have unit trace: $\tr(\rho) = 1$.
	
	Following~\cite{selinger2004towards}, we adopt a slightly more general notion of quantum states that proves convenient for modeling non-terminating quantum programs. We work with \emph{partial density operators}—positive operators with trace at most 1—as our representation of (possibly unnormalized) quantum states. 
	The intuition behind this generalization is as follows. When $\tr(\rho) \neq 0$, the partial density operator $\rho$ represents a legitimate quantum state $\rho / \tr(\rho)$ that occurs with probability $\tr(\rho)$. When $\tr(\rho) = 0$, which necessarily implies $\rho = 0$ since $\rho$ is positive, it does not correspond to any physical quantum state; in our framework, this represents non-termination, indicating that no quantum state has been reached by the program.
	
	Let $\l(\h)$, $\d(\h)$, and $\p(\h)$ denote the sets of linear operators, partial density operators, and \emph{effects} (positive operators with eigenvalues in $[0,1]$) on $\h$, respectively. We have the inclusions $\d(\h) \subseteq \p(\h) \subseteq \l(\h)$. 
	The \emph{L\"{o}wner order} $\le$ on $\l(\h)$ is defined as follows: $A \le B$ iff $\ B - A$ is positive. Under this order, both $(\d(\h), \le)$ and $(\p(\h), \le)$ form complete partial orders (CPOs) with the zero operator $0$ as the minimum element. Additionally, $\p(\h)$ has the identity operator $I_\h$ as its maximum element.

	Composite quantum systems, such as multi-qubit systems, are described by tensor product state spaces. If subsystems have state spaces $\h_1$ and $\h_2$, the composite system has state space $\h_1 \otimes \h_2$. 
	
	For any operator $\rho \in \l(\h_1 \otimes \h_2)$, we can compute its \emph{reduced states} on individual subsystems via the \emph{partial trace} operations $\tr_{\h_1}$ and $\tr_{\h_2}$. The partial trace $\tr_{\h_1}: \l(\h_1 \otimes \h_2) \to \l(\h_2)$ is the unique linear map satisfying
	\[
	\tr_{\h_1}(|\psi_1\rangle\langle\phi_1| \otimes |\psi_2\rangle\langle\phi_2|) =
\tr(|\psi_1\rangle\langle\phi_1|)\cdot |\psi_2\rangle\langle\phi_2|= \langle\phi_1|\psi_1\rangle \cdot |\psi_2\rangle\langle\phi_2|,
	\]
	where $\langle\phi_1|\psi_1\rangle$ is the inner product of $|\phi_1\rangle$ and $|\psi_1\rangle$, and $\cdot$ denotes scalar multiplication. The partial trace $\tr_{\h_2}$ is defined analogously. When $\rho \in \d(\h_1 \otimes \h_2)$ represents a joint quantum state, $\tr_{\h_1}(\rho)$ and $\tr_{\h_2}(\rho)$ give the marginal states on subsystems $\h_2$ and $\h_1$, respectively.
	
	\subsubsection{Quantum Operations}
	
	The dynamics of an isolated quantum system are governed by \emph{unitary evolution}. If a system's state at time $t_1$ is $\rho_1$ and at time $t_2$ is $\rho_2$, then $\rho_2 = U \rho_1 U^\dagger$, where $U$ is a unitary operator (satisfying $U^\dagger U = U U^\dagger = I_\h$) and $U^\dagger$ denotes its conjugate transpose. 
	Standard single-qubit unitary operators include the Pauli matrices $X$, $Y$, $Z$, the phase gate $S$, and the Hadamard gate $H$. For two-qubit systems, the controlled-NOT gate $CX$ is fundamental. These are represented in matrix form as
	\[
X\define \begin{bmatrix}
	0 & 1 \\
	1 & 0 
\end{bmatrix},\ Y\define \begin{bmatrix}
	0 & -i \\
	i & 0 
\end{bmatrix},\ Z\define \begin{bmatrix}
	1 & 0 \\
	0 & -1
\end{bmatrix},\ S\define \begin{bmatrix}
	1 & 0 \\
	0 & i
\end{bmatrix},\ H\define \frac{1}{\sqrt{2}} \begin{bmatrix}
	1 & 1 \\
	1 & -1
\end{bmatrix},\ CX\define \begin{bmatrix}
	1 & 0 & 0 & 0 \\
	0 & 1 & 0 & 0 \\
	0 & 0 & 0 & 1 \\
	0 & 0 & 1 & 0
\end{bmatrix}
.\]
	
	To extract classical information from a quantum system, we perform a \emph{measurement}. A projective measurement $\m$ is mathematically described by a set $\{P_i : i \in O\}$ of projectors (Hermitian operators with eigenvalues in $\{0,1\}$), where $O$ is the set of possible measurement outcomes. These projectors must satisfy the \emph{completeness relation} $
	\sum_{i \in O} P_i = I_\h.$
	The measurement process is inherently probabilistic. If the system is initially in state $\rho$, then the probability of observing outcome $i$ is $p_i = \tr(P_i \rho)$, and conditioned on observing outcome $i$ (when $p_i > 0$), the post-measurement state collapses to $\rho_i = P_i \rho P_i / p_i$.

	A quantum measurement can equivalently be specified by a Hermitian operator $M \in \l(\h)$ called an \emph{observable}. Through spectral decomposition, we can write
	$
	M = \sum_{m \in \mathit{spec}(M)} m \cdot P_m,
	$
	where $\mathit{spec}(M)$ denotes the set of eigenvalues of $M$, and $P_m$ is the projector onto the eigenspace associated with eigenvalue $m$. The corresponding projective measurement is then $\{P_m : m \in \mathit{spec}(M)\}$. 
	By linearity of the trace, the expected value of measurement outcomes when measuring observable $M$ on state $\rho$ is simply
	\[
	\sum_{m \in \mathit{spec}(M)} m \cdot \tr(P_m \rho) = \tr(M \rho).
	\]
	This formula provides a convenient way to compute expectation values without explicitly referencing the underlying projective measurement.

	The general evolution of quantum systems, including those that interact with external environments or undergo measurements, is described by \emph{super-operators}. A super-operator is a linear map $\e: \l(\h_1) \to \l(\h_2)$ that transforms operators on one Hilbert space to operators on another. However, not all linear maps between operator spaces correspond to physically realizable quantum operations. To ensure physical consistency, we require super-operators to satisfy certain positivity conditions:
	
	\begin{enumerate}
		\item $\e$ is \emph{positive} if it maps positive operators to positive operators. This requirement is necessary because density operators are positive.
		
		\item More stringently, $\e$ is \emph{completely positive} if the extended map $\mathcal{I}_\h \otimes \e: \l(\h \otimes \h_1) \to \l(\h \otimes \h_2)$ is positive for every finite-dimensional Hilbert space $\h$, where
		\[
		(\mathcal{I}_\h \otimes \e)(A \otimes B) = A \otimes \e(B),
		\]
		and $\mathcal{I}_\h$ denotes the identity super-operator on $\l(\h)$.
	Complete positivity is essential because quantum systems may be entangled with external systems, and the requirement ensures that the operation remains physical even when applied to part of a larger entangled state.

		\item $\e$ is  \emph{trace-preserving} (resp. \emph{trace-nonincreasing}) if 
		$\tr(\e(A)) = \tr(A)$ (resp. $\tr(\e(A)) \leq \tr(A)$) for any positive operator $A\in \l(\h_1)$. Again, this requirement is necessary because density operators have unit trace (or have trace at most 1 for partial density operators).
		
	\end{enumerate}
	
	Completely positive and trace-preserving (CPTP) super-operators are also called \emph{quantum channels} in quantum information theory. Completely positive and trace-nonincreasing (CPTN) super-operators are more general and naturally arise when modeling quantum programs that may not terminate. Standard quantum operations that are CPTP include \emph{unitary transformation} $\e_U(\rho) = U\rho U^\dagger$ and \emph{complete measurement} $\e_\m(\rho) = \sum_{i\in O} P_i \rho P_i$, where $\m = \{P_i : i\in O\}$ is a projective measurement. This represents the post-measurement state when all measurement outcomes are accounted for but the classical outcome is discarded.
	In contrast, operations $
	\e_i(\rho) = P_i \rho P_i,
	$ corresponding to individual measurement outcomes, are typically CPTN but not CPTP,
	where $\e_i(\rho)$ is a partial density operator and the trace $\tr(\e_i(\rho)) = \tr(P_i \rho)$ equals the probability of observing outcome $i$.

	A fundamental result in quantum computing is the \emph{Kraus representation theorem}~\cite{kraus1983states}, which provides an operational characterization of completely positive maps. To be specific, a super-operator $\e: \l(\h_1) \to \l(\h_2)$ is completely positive iff there exist linear operators $\{E_i : i \in I\}$, called \emph{Kraus operators}, mapping $\h_1$ to $\h_2$ such that
	\[
	\e(A) = \sum_{i\in I} E_i A E_i^\dagger \quad \text{for all } A \in \l(\h_1).
	\]
	Furthermore, $\e$ is trace-preserving (resp. {trace-nonincreasing}) iff $\ \sum_{i\in I} E_i^\dag E_i = I_{\h_1}$ (resp. $\sum_{i\in I} E_i^\dag E_i \le I_{\h_1}$). 
	It is worth noting that the Kraus representation is not unique; that is, different sets of Kraus operators can represent the same super-operator. 
	
	Every super-operator $\e: \l(\h_1) \to \l(\h_2)$ has an \emph{adjoint} $\e^\dagger: \l(\h_2) \to \l(\h_1)$ uniquely characterized by the property
	\[
	\tr(\e(A) B) = \tr(A \e^\dagger(B)) \quad \text{for all } A \in \l(\h_1), B \in \l(\h_2).
	\]
	When $\e$ has Kraus representation $\e(A) = \sum_{i\in I} E_i A E_i^\dagger$, its adjoint has Kraus operators $\{E_i^\dagger : i \in I\}$, giving
	$
	\e^\dagger(B) = \sum_{i\in I} E_i^\dagger B E_i
	$.
	It follows that $(\e^\dagger)^\dagger = \e$, and if $\e$ is completely positive, then so is $\e^\dagger$. The adjoint operation will play an important role in our characterization of refinement orders.
	
	\subsection{Quantum Variables and Notation}
	
	To facilitate discussion of quantum programs, we introduce quantum variables. Let $\QVar$ be a countably infinite set of \emph{quantum variables}, denoted by $q, q', q_1, \ldots$, where each variable represents a qubit (i.e., has a 2-dimensional state space, say $\h_q$). We use $V, W, \ldots$ to denote finite subsets of $\QVar$.
	
	The Hilbert space associated with a set $V$ of quantum variables is the tensor product:
	\[
	\h_V = \bigotimes_{q\in V} \h_q.
	\]
	We write $V' \| V$ to indicate that $V'$ is a \emph{disjoint copy} of $V$; that is, $|V'| = |V|$ and $V \cap V' = \emptyset$. This notation will be useful when discussing program transformations applying on a maximally entangled state.
	
	\subsection{Domain-Theoretic Orders on Powersets}
	
	In the study of nondeterministic programs, we model nondeterminism using sets of possible behaviors. To reason about refinement of such sets, we require orderings on powersets that generalize the underlying order on individual elements. Domain theory~\cite{abramsky1994domain} provides three standard such orderings: the Hoare, Smyth, and Egli-Milner orders.
	
	Let $(\mathcal{X}, \leq)$ be a partially ordered set. For subsets $X, Y \subseteq \mathcal{X}$, we define
	\begin{align*}
		X \leq_H Y &\quad \text{if} \quad \forall x \in X,\ \exists y \in Y: x \leq y \quad \text{(Hoare order)}\\
		X \leq_S Y &\quad \text{if} \quad \forall y \in Y,\ \exists x \in X: x \leq y \quad \text{(Smyth order)}\\
		X \leq_{EM} Y &\quad \text{if} \quad X \leq_H Y\ \text{ and }\ X \leq_S Y \quad \text{(Egli-Milner order)}
	\end{align*}
	The Hoare order $X \leq_H Y$ requires that every element of $X$ is upper bounded by  some element of $Y$, while the Smyth order $X \leq_S Y$ requires that every element of $Y$ can be lower bounded by some element of $X$. The Egli-Milner order combines both requirements.
	
	These orders can also be characterized using closure operators. Define the \emph{up-closure} and \emph{down-closure} of $X$ as
	\begin{align*}
		\upcl X &= \{y \in \mathcal{X} : \exists x \in X, x \leq y\}\\
		\downcl X &= \{y \in \mathcal{X} : \exists x \in X, y \leq x\}
	\end{align*}
	Then we have the equivalent characterizations:
	\begin{align*}
		X \leq_H Y &\quad \text{iff} \quad \downcl X \subseteq \downcl Y\\
		X \leq_S Y &\quad \text{iff} \quad \upcl Y \subseteq \upcl X
	\end{align*}
	When $X = \{x\}$ and $Y = \{y\}$ are both singletons, all three powerset orders reduce to the underlying order on $\mathcal{X}$. That is, 
	$
	X \leq_* Y$ iff $x \leq y$  for $* \in \{H, S, EM\}.
	$
	This ensures that the powerset orders are genuine generalizations of the base order, which will be important when connecting refinement orders for deterministic and nondeterministic quantum programs.
	
	\section{State Predicates for Quantum Programs}
	\label{sec:predicate}
	
	In classical program verification, predicates are typically Boolean-valued functions that describe properties of program states. The quantum setting, however, offers multiple mathematically valid approaches to defining state predicates, each with different expressive power and applications. This section systematically introduces three classes of quantum predicates that have been studied in the literature: projectors~\cite{zhou2019applied, unruh2019quantum, unruh2019quantumr, feng2023abstract}, effects~\cite{d2006quantum, ying2012floyd,ying2024foundations},  and sets of effects~\cite{feng2023verification}. We examine their mathematical structures, satisfaction relations, and the ordering relationships between them, laying foundations that will be essential for our analysis of refinement orders in subsequent sections.

	\subsection{Projectors: Qualitative Properties}
	
	Projectors provide the most restrictive form of quantum predicates, capturing qualitative, yes-or-no properties of quantum states. For a set $V \subseteq \QVar$ of quantum variables, let $\s(\h_{V})$ denote the set of \emph{projectors} on the Hilbert space $\h_{V}$. Recall that a projector is a Hermitian operator $P$ satisfying $P^2 = P$, which equivalently has eigenvalues in $\{0, 1\}$.
	
	\paragraph{Lattice structure.} There is a well-known one-to-one correspondence between projectors and closed subspaces of $\h_{V}$: each projector $P$ corresponds to its image (the eigenspace associated with eigenvalue 1). Throughout this paper, we freely identify projectors with their corresponding subspaces for notational convenience. Under this identification, the set of projectors forms a complete lattice
	\[
	\left(\s(\h_{V}), \le, \wedge, \vee, \bot = 0, \top = I_{V}\right),
	\]
	where the L\"{o}wner order $\le$ on projectors corresponds to subspace inclusion: $P \le Q$ as projectors iff $\ P \subseteq Q$ as subspaces. The meet $P \wedge Q$ corresponds to subspace intersection $P \cap Q$, while the join $P \vee Q$ corresponds to the smallest subspace containing both $P$ and $Q$ (the closed linear span $P + Q$).
	Finally, the bottom element $\bot = 0$ is the zero subspace, and the top element $\top = I_V$ is the entire space $\h_V$.
	
	\paragraph{Satisfaction relation.} Following~\cite{zhou2019applied}, we use projectors to express qualitative properties of quantum states. For a state $\rho \in \d(\h_{V})$ and a projector $P \in \s(\h_{V})$, we define
	\[
	\rho \models P \quad \text{if} \quad \supp{\rho} \subseteq P,
	\]
	where $\supp{\rho}$ denotes the \emph{support subspace} of $\rho$—the subspace spanned by eigenvectors of $\rho$ corresponding to nonzero eigenvalues. Equivalently, the support can be characterized as the orthogonal complement of the null subspace; that is,
	\[
	\supp{\rho} = \mathcal{N}(\rho)^\bot, \quad \text{where} \quad \mathcal{N}(\rho) = \left\{|\psi\rangle \in \h_{V} : \langle \psi | \rho | \psi \rangle = 0\right\}.
	\]
	
	The satisfaction relation $\rho \models P$ has a clear physical meaning. Since $\rho \models P$ iff $\ \tr(P\rho) = \tr(\rho)$, this condition is equivalent to the following operational criterion: performing the binary projective measurement $\{P, I_V - P\}$ on state $\rho$ will yield outcome $P$ with certainty (probability 1). This provides an experimentally verifiable interpretation of projector satisfaction.
	
	\begin{example}\label{exm:proj} Consider a single qubit with computational basis states $|0\rangle$ and $|1\rangle$. The projector $P = |0\rangle\langle0|$ represents the property ``the qubit is in the state $|0\rangle$.'' A pure state $\rho = |0\rangle\langle0|$ satisfies $P$ because $\supp{\rho} = \text{span}\{|0\rangle\} = P$. In contrast, the orthogonal state $\rho' = |1\rangle\langle1|$ does not satisfy $P$ since $\supp{\rho'} = \text{span}\{|1\rangle\}$ is not included in $P$. More generally, any superposition $\rho = \alpha|0\rangle\langle0| + \beta|1\rangle\langle1|$ with $\beta \neq 0$ fails to satisfy $P$.
	\end{example}
	
	\subsection{Effects: Quantitative Properties}
	
	While projectors capture binary properties, many quantum phenomena require quantitative descriptions. Effects provide this capability by allowing predicates to express probabilistic and graded properties of quantum states.
	
	\paragraph{CPO structure.}  For any finite $V \subseteq \QVar$, recall that $\p(\h_{V})$ denotes the set of all effects on $\h_V$, where an {effect} is a positive  operator with eigenvalues in the interval $[0, 1]$. Unlike projectors, effects do not form a lattice—the meet and join of two effects generally fail to exist unless they commute. However, $\p(\h_{V})$ does possess important structure as a partially ordered set:
	\[
	\left(\p(\h_{V}), \le, \bot = 0, \top = I_{V}\right)
	\]
	is a complete partial order (CPO) under the L\"{o}wner order $\le$, with minimum element $0$ and maximum element $I_V$. Note that $\s(\h_V) \subseteq \p(\h_{V})$, so projectors form a special subclass of effects corresponding to $\{0, 1\}$-valued observables.
	
	\paragraph{Graded satisfaction.} Given a quantum state $\rho \in \d(\h_{V})$ and an effect $M \in \p(\h_{V})$, the degree to which $\rho$ satisfies $M$ is defined as
	\[
	\Exp(\rho \models M) = \tr(M\rho).
	\]
 Recall from Section~\ref{sec:preliminaries} that $\tr(M\rho)$ is the expected value of measurement outcomes when measuring observable $M$ on state $\rho$. This gives a natural operational meaning of $\Exp(\rho \models M)$.
	
	\paragraph{Connection to projectors.} For any projector $P \in \s(\h_V)$, we have
	\[
	\rho \models P \quad \text{iff} \quad \Exp(\rho \models P) = \tr(\rho).
	\]
	This shows that projector satisfaction is the special case where the expected value equals the maximum possible value, recovering the qualitative notion from the quantitative framework.

	\begin{example} 		
		Take $P = |0\rangle\langle 0|$ as in Example~\ref{exm:proj}. Regarding as an effect, it computes the probability of satisfying 
		the property ``the qubit is in the state $|0\rangle$.'' To be specific, for a mixed state $\rho$, the expression $\tr(P\rho)$ computes precisely the probability that $P$ holds at $\rho$.
		
	\end{example}

	\subsection{Sets of Effects: Nondeterministic Specifications}
	
	To handle nondeterminism in quantum programs and develop expressive verification systems, it proves useful to consider sets of effects as predicates~\cite{feng2023verification}. This generalization enables specifications that capture ranges of possible behaviors and both demonic (adversarial) and angelic (benevolent) nondeterminism.
	
		\paragraph{Power domain structure.} For any finite $V \subseteq \QVar$, we work within the powerset $2^{\p(\h_V)}$ of all subsets of effects. On this collection, the \emph{Hoare order} $\leq_H$ and the \emph{Smyth order} $\leq_S$ from domain theory can be naturally defined based on the base L\"{o}wner order $\le$.
	For $\qassert, \qassertp \in 2^{\p(\h_V)}$,
\begin{align*}
	\qassert \leq_H \qassertp \quad &\text{iff} \quad \forall M \in \qassert, \, \exists N \in \qassertp: M \le N,
	\\
	\qassert \leq_S \qassertp \quad &\text{iff} \quad \forall N \in \qassertp, \, \exists M \in \qassert: M \le N.	
	\end{align*}
	
	\paragraph{Satisfaction Relations.} The meaning of a specification is given by when a quantum state satisfies it. We define two complementary satisfaction relations, reflecting different attitudes towards the nondeterminism within a set $\qassert \subseteq \p(\h_V)$.
	The \emph{demonic (pessimistic) satisfaction} quantifies the \emph{guaranteed} degree of satisfaction, assuming an adversarial environment that chooses the worst possible effect from the set:
	\[
	\Exp_{\mathsf{dem}}(\rho \models \qassert) = \inf_{M \in \qassert} \tr(M\rho).
	\]
	By convention, $\inf_{M \in \emptyset} \tr(M\rho) = \tr(\rho)$, meaning an empty specification is trivially maximally satisfied.	
	Dually, the \emph{angelic (optimistic) satisfaction} quantifies the \emph{possible} degree of satisfaction, assuming a benevolent environment that chooses the best possible effect:
	\[
	\Exp_{\mathsf{ang}}(\rho \models \qassert) = \sup_{M \in \qassert} \tr(M\rho).
	\]
	By convention, $\sup_{M \in \emptyset} \tr(M\rho) = 0$.
	
	These satisfaction relations induce natural refinement orders between sets of effects, comparing their behavioral guarantees across all states.
	The \emph{demonic order} $\preceq_{\mathsf{dem}}$ requires that the guaranteed satisfaction of $\qassert$ is never better than that of $\qassertp$:
	\[
	\qassert \preceq_{\mathsf{dem}} \qassertp \quad \text{if} \quad \forall \rho \in \d(\h_V): \; \Exp_{\mathsf{dem}}(\rho \models \qassert) \leq \Exp_{\mathsf{dem}}(\rho \models \qassertp).
	\]
	In contrast, the \emph{angelic order} $\preceq_{\mathsf{ang}}$ requires that the possible satisfaction of $\qassert$ is never better than that of $\qassertp$:
	\[
	\qassert \preceq_{\mathsf{ang}} \qassertp \quad \text{if} \quad \forall \rho \in \d(\h_V): \; \Exp_{\mathsf{ang}}(\rho \models \qassert) \leq \Exp_{\mathsf{ang}}(\rho \models \qassertp).
	\]
	
	\subsection{Relating Satisfaction-Based and Domain-Theoretic Orders}
	
	We are going to show an intimate connection between the satisfaction-based orders ($\preceq_{\mathsf{dem}}$, $\preceq_{\mathsf{ang}}$) and the domain-theoretic orders ($\leq_S$, $\leq_H$). 
	To this end, we first present the following lemma.
	\begin{lemma}[Duality]\label{lem:duality}
		For any predicates $\qassert, \qassertp \in 2^{\p(\h_V)}$, the following equivalences hold:
		\begin{enumerate}
			\item $\qassert \preceq_{\mathsf{dem}} \qassertp$ iff $\ (I - \qassertp) \preceq_{\mathsf{ang}} (I - \qassert)$;
			\item $\qassert \leq_S \qassertp$ iff $\ (I - \qassertp) \leq_H (I - \qassert)$.
		\end{enumerate}
		Here, $I - \qassert$ denotes the set $\{I - M : M \in \qassert\}$.
	\end{lemma}
	\begin{proof}
		For (1), we compute 
		\begin{align*}
			\qassert \preceq_{\mathsf{dem}} \qassertp
			&\quad\mbox{iff}\quad \forall \rho: \inf_{M \in \qassert} \tr(M\rho) \leq \inf_{N \in \qassertp} \tr(N\rho) \quad \text{(definition of $ \preceq_{\mathsf{dem}} $)}\\
			&\quad\mbox{iff}\quad \forall \rho: \sup_{M \in \qassert} \tr((I - M)\rho) \geq \sup_{N \in \qassertp} \tr((I - N)\rho) \\
			&\qquad\qquad\qquad\qquad\qquad \text{(since $\tr((I - M)\rho) = \tr(\rho) - \tr(M\rho)$)}\\
			&\quad\mbox{iff}\quad (I - \qassertp) \preceq_{\mathsf{ang}} (I - \qassert)\quad \text{(definition of $\preceq_{\mathsf{ang}}$)}. 
		\end{align*}
		For (2), by definition of the Smyth order, $\qassert \leq_S \qassertp$ means $\forall N \in \qassertp, \exists M \in \qassert: M \le N$. This is equivalent to $\forall N \in \qassertp, \exists M \in \qassert: I - N \le I - M$ (by properties of the L\"{o}wner order), which precisely states that $(I - \qassertp) \leq_H (I - \qassert)$ by definition of the Hoare order.
	\end{proof}

	For well-behaved (convex and closed) specifications, the satisfaction-based and domain-theoretic orders coincide. This equivalence relies on the Sion minimax theorem.
	
	\begin{theorem}[Sion Minimax Theorem~\cite{sion1958general}]\label{thm:sm}
		Let $X$ and $Y$ be non-empty, convex, compact subsets of two linear topological spaces, and let $f: X \times Y \to \mathbb{R}$ be a function that is linear in the first argument and convex in the second. Then
		\[
		\min_{x \in X} \max_{y \in Y} f(x, y) = \max_{y \in Y} \min_{x \in X} f(x, y).
		\]
	\end{theorem}
	
	\begin{theorem}\label{thm:equivalence}
		For any convex and closed predicates $\qassert, \qassertp \in 2^{\p(\h_V)}$,
		\begin{enumerate}
			\item if $\ \qassert \neq \emptyset$, then $\qassert \preceq_{\mathsf{dem}} \qassertp$ iff $\ \qassert \leq_S \qassertp$;
			\item if $\ \qassertp \neq \emptyset$, then $\qassert \preceq_{\mathsf{ang}} \qassertp$ iff $\ \qassert \leq_H \qassertp$.
		\end{enumerate}
	\end{theorem}
	\begin{proof}
		We prove part (1); part (2) follows by applying part (1) to the complementary sets using Lemma~\ref{lem:duality}.

First, observe that when $\qassert \neq \emptyset$ and both predicates are closed (so that infima are attained), we can rewrite the definitions as:
\begin{align*}
	\qassert \leinf \qassertp
	&\quad\mbox{iff}\quad \forall \rho \in \d(\h_V), \forall N \in \qassertp, \exists M \in \qassert: \tr(M\rho) \leq \tr(N\rho)\\
	\qassert \leq_S \qassertp
	&\quad\mbox{iff}\quad \forall N \in \qassertp, \exists M \in \qassert, \forall \rho \in \d(\h_V): \tr(M\rho) \leq \tr(N\rho).
\end{align*}
The implication $\qassert \leq_S \qassertp \Rightarrow \qassert \leinf \qassertp$ is immediate.
For the reverse implication, assume $\qassert \leinf \qassertp$ and fix any $N \in \qassertp$. We need to find $M \in \qassert$ such that $M \le N$, which is equivalent to $\tr(M\rho) \leq \tr(N\rho)$ for all $\rho \in \d(\h_V)$.

Since $\h_V$ is finite-dimensional, both $\d(\h_V)$ and $\qassert$ are non-empty, convex, and compact subsets of appropriate linear topological spaces. Moreover, the function $f(\rho, M) = \tr((N - M)\rho)$ is linear in $\rho$ and convex in $M$. By the Sion minimax theorem,
\begin{align*}
	\qassert \leinf \qassertp
	&\Rightarrow \forall N \in \qassertp, \min_{\rho \in \d(\h_V)} \max_{M \in \qassert} \tr[(N - M)\rho] \geq 0
	\quad \text{(by assumption)}\\
	&\Rightarrow \forall N \in \qassertp, \max_{M \in \qassert} \min_{\rho \in \d(\h_V)} \tr[(N - M)\rho] \geq 0
	\quad \text{(Sion minimax)}\\
	&\Rightarrow \forall N \in \qassertp, \exists M \in \qassert: \min_{\rho \in \d(\h_V)} \tr[(N - M)\rho] \geq 0\\
	&\Rightarrow \forall N \in \qassertp, \exists M \in \qassert, \forall \rho \in \d(\h_V): \tr(M\rho) \leq \tr(N\rho)\\
	&\Rightarrow \qassert \leq_S \qassertp \quad \text{(definition of $\leq_S$)}. \qedhere
\end{align*}
\end{proof}

\paragraph{Necessity of side conditions.} The non-emptiness conditions in Theorem~\ref{thm:equivalence} are essential. For part (1), note that $\emptyset \leinf \{I\}$ (since the infimum over the empty set is $\tr(\rho)$ by convention), yet $\emptyset \not\leq_S \{I\}$ (the Smyth order requires finding an element in the left-hand set). Similarly, for part (2), we have $\{0\} \lesup \emptyset$ (since both sides equal 0) but $\{0\} \not\leq_H \emptyset$ (the Hoare order cannot be satisfied when the right-hand set is empty).

	\section{Refinement for Deterministic Quantum Programs}
	\label{sec:deterministic}
	
	This section establishes refinement orders for deterministic quantum programs. We begin by introducing our semantic model—completely positive trace-nonincreasing (CPTN) super-operators, and the natural approximation ordering on this space. We then characterize refinement orders under effect-based specifications, showing they coincide precisely with the approximation ordering. Finally, we demonstrate that restricting to projector-based specifications yields strictly weaker refinement orders.
	
	\subsection{Semantic Model and Approximation Ordering}
	
	The most general mathematical representation of a physically realizable quantum process is a completely positive and trace-nonincreasing (CPTN) super-operator. This generality is not merely theoretical: the quantum while-language studied in~\cite{ying2012floyd, chadha2006reasoning, Kakutani:2009, unruh2019quantum, feng2007proof, rand2019verification} can implement any CPTN super-operator on finite-dimensional Hilbert spaces~\cite{ying2024foundations}. This justifies our semantic approach of identifying deterministic quantum programs with CPTN super-operators.
	
	For a finite set $V \subseteq \QVar$ of quantum variables, we define
	\[
	\dprog(V) = \left\{\e: \d(\h_V) \to \d(\h_V) \mid \e \text{ is a CPTN super-operator}\right\}
	\]
	as the semantic space of deterministic quantum programs operating on system $V$. We equip $\dprog(V)$ with the following \emph{approximation order}:
	\[
	\e \le \f \quad \text{if} \quad \f - \e \text{ is completely positive}.
	\]
	Note that we use the same symbol $\le$ as for the L\"{o}wner order on operators, but the semantics differ: for super-operators, we require complete positivity of the difference, not merely positivity. This ordering makes $\dprog(V)$ a complete partial order (CPO) with the zero super-operator $0_V$ as the minimum element.

	\begin{remark} An alternative partial order on $\dprog(V)$ can be defined by pointwise lifting of the L\"{o}wner order as follows: $
	\e \leq^L \f$ if $\forall \rho \in \d(\h_V): \e(\rho) \le \f(\rho).
	$
	This pointwise order, adopted in~\cite{ying2024foundations}, also makes $\dprog(V)$ a CPO and provides fixed-point semantics for while loops. However, we argue that the approximation ordering $\le$ is more natural for program reasoning.
	The key issue is \emph{compositional behavior}. When a program $\e \in \dprog(V)$ is used as a subprogram within a larger system, it must be extended trivially to $\id_{W\backslash V} \otimes \e$ for some $W \supseteq V$. For a compositional order, we require
	\[
	\e \le \f \quad \text{iff} \quad \id_{W\backslash V} \otimes \e \le \id_{W\backslash V} \otimes \f.
	\]
	This property holds for the approximation ordering $\le$ but fails for the pointwise order $\leq^L$. To see this, consider the transpose super-operator $\mathcal{T}$ defined by $\mathcal{T}(A) = A^T$ (matrix transpose). We have $0_V \leq^L \mathcal{T}$ since $0 \le A^T$ for all positive $A$. However, for any disjoint copy $V'$ with $V' \| V$, we have $0_{V' \cup V} \not\leq^L \id_{V'} \otimes \mathcal{T}$~\cite{nielsen2002quantum}. This failure of compositionality demonstrates why $\le$ is preferable to $\leq^L$ for program refinement.
	
	\end{remark}
	\subsection{Choi-Jamiolkowski Representation}
	
	To work effectively with the approximation ordering, we employ the Choi-Jamiolkowski isomorphism, which provides a bridge between properties of super-operators and properties of operators.
	
	 Let $\e$ be a super-operator on $\l(\h_V)$ and let $V'$ be a disjoint copy of $V$. Define the \emph{maximally entangled state} on $\h_{V,V'} = \h_V \otimes \h_{V'}$ by
	\[
	\Omega = |\omega\rangle\langle\omega|, \quad \text{where} \quad |\omega\rangle = \frac{1}{\sqrt{d}}\sum_{i=0}^{d-1} |i\rangle_V |i\rangle_{V'},
	\]
	with $d = \dim(\h_V) = 2^{|V|}$, and $\{|i\rangle_W : 0 \leq i < d\}$ denotes an orthonormal basis of $\h_W$ for $W \in \{V, V'\}$.
	The \emph{Choi-Jamiolkowski matrix} of $\e$ is defined as
	\[
	J_{V'}(\e) = \e(\Omega),
	\]
	where $\e(\Omega)$ is understood as $(\e \otimes \id_{V'})(\Omega)$. Since the properties of $J_{V'}(\e)$ does not depend on the choice of $V'$ as long as $|V'| = |V|$ and $V\cap V' =\emptyset$, we omit the subscript $V'$ and write $J(\e)$ for simplicity.  
	
	A fundamental result states that $\e$ is completely positive iff $\ J(\e)$ is positive~\cite{nielsen2002quantum}. 
	The following lemma establishes several equivalent characterizations of the approximation ordering, which will be essential for our subsequent analysis.
	
	\begin{lemma}[Characterizations of Approximation Order]\label{lem:choi}
		For any two deterministic quantum programs $\e, \f \in \dprog(V)$, the following are equivalent:
		\begin{enumerate}
			\item $\e \le \f$ (approximation order);
			\item $J(\e) \le J(\f)$ (L\"{o}wner order on Choi matrices);
			\item For any $W \supseteq V$ and $A \in \p(\h_W)$, $\e(A) \le \f(A)$;
			\item For any $V' \| V$ and $A \in \p(\h_{V \cup V'})$, $\e(A) \le \f(A)$;
			\item $\e^\dagger \le \f^\dagger$ (approximation order for adjoints).
		\end{enumerate}
	\end{lemma}
	
	\begin{proof}
		The equivalence of (1) and (2) follows from the Choi-Jamiolkowski isomorphism: $\e \le \f$ iff $\f - \e$ is completely positive, which holds iff $J(\f - \e) = J(\f) - J(\e)$ is positive, i.e., $J(\e) \le J(\f)$.
		
		That (3) or (4) implies (2) is immediate by taking $A = \Omega$. That (1) implies (3) and (4) follows from the definition of complete positivity applied to $\f - \e$. Finally, for the equivalence of (4) and (5), we compute
		\begin{align*}
			&\forall A \in \p(\h_{V \cup V'}): \e(A) \le \f(A)\\
			&\quad\mbox{iff}\quad\forall A, N \in \p(\h_{V \cup V'}): \tr(N \e(A)) \leq \tr(N \f(A))
			\quad \text{(L\"{o}wner order characterization)}\\
			&\quad\mbox{iff}\quad\forall N, A \in \p(\h_{V \cup V'}): \tr(\e^\dagger(N) A) \leq \tr(\f^\dagger(N) A)
			\quad \text{(adjoint definition)}\\
			&\quad\mbox{iff}\quad\forall N \in \p(\h_{V \cup V'}): \e^\dagger(N) \le \f^\dagger(N)
			\quad \text{(L\"{o}wner order characterization)}\\
			&\quad\mbox{iff}\quad\e^\dagger \le \f^\dagger
			\quad   \text{(equivalence of (1) and (4))}. \qedhere
		\end{align*}
	\end{proof}
	
	This lemma provides multiple computational tools for verifying the approximation order. Clause (2) is particularly useful as it reduces super-operator comparison to operator comparison. Clause (5) shows that the approximation order is self-dual with respect to adjoints.
	
	\subsection{Refinement Under Effect-Based Specifications}
	
	We now define refinement orders based on effect predicates, following the satisfaction framework introduced in the previous section. For a finite set $W \subseteq \QVar$, define
	\[
	\dspec(W) = \left\{(M, N) : M, N \in \p(\h_W)\right\}
	\]
	as the set of effect-based specifications on system $W$, where $M$ represents the precondition and $N$ the postcondition.
	
	\subsubsection{Correctness definitions}
	
	 Depending on different views for non-termination, there are two different notions of satisfaction of a quantum program $\e$ in $\dprog(V)$ on a specification $(M,N)$ in $\dspec(W)$, one is for total correctness and the other for partial correctness. Specifically, 
	 
	 \begin{enumerate}
	 	\item 
We say $\e$ satisfies $(M, N)$ in the sense of \emph{total correctness}, denoted $\e \dtsat (M, N)$, if for any quantum state $\rho \in \d(\h_X)$ with $V \cup W \subseteq X$,
	\[
	\tr(M\rho) \leq \tr(N\e(\rho)).
	\]
	Note that we do not require $W=V$ in general. However, say, $\tr(N\e(\rho))$ is understood as $\tr[(I_{X\backslash W} \otimes N) \cdot (\id_{X\backslash V} \otimes \e)(\rho)]$, where programs and predicates are extended to the full system $X$ via tensor products with identity operators. We adopt this simplified notation for clarity.
	
	Intuitively, total correctness requires that whenever the precondition $M$ is satisfied by the initial state $\rho$ to degree $\tr(M\rho)$, the postcondition $N$ must be satisfied to at least the same degree by the final state $\e(\rho)$.
	
	\item We say $\e$ satisfies $(M, N)$ in the sense of \emph{partial correctness}, denoted $\e \dpsat (M, N)$, if for any $\rho \in \d(\h_X)$ with $V \cup W \subseteq X$,
	\[
	\tr(M\rho) \leq \tr(N\e(\rho)) + \tr(\rho) - \tr(\e(\rho)).
	\]
	The additional term $\tr(\rho) - \tr(\e(\rho))$ accounts for non-termination probability. 
\end{enumerate}
	
	These satisfaction relations can also be characterized using predicate transformers. As shown in~\cite{d2006quantum}, the \emph{weakest precondition} of $\e$ with respect to postcondition $N$ is $
	wp.\e.N = \e^\dagger(N),$
	and the \emph{weakest liberal precondition} is $
	wlp.\e.N = I - \e^\dagger(I - N).$
	It is easy to check that $$\e \dtsat (M, N)\quad \mbox{iff}\quad M \le wp.\e.N \qquad\mbox{and}\qquad \e \dpsat (M, N)\quad \mbox{iff}\quad M \le wlp.\e.N.$$
	
	\subsubsection{Refinement Orders}
	
	With satisfaction relations defined, we now introduce refinement orders. For any $\e, \f \in \dprog(V)$,
	
	\begin{enumerate}
		\item $\e$ {is refined by $\f$ in total correctness}, denoted $\e \dtrefine \f$, if for every $W \subseteq \QVar$ and $(M, N) \in \dspec(W)$,
		\[
		\e \dtsat (M, N) \quad \text{implies} \quad \f \dtsat (M, N).
		\]
		
		\item $\e$ {is refined by $\f$ in partial correctness}, denoted $\e \dprefine \f$, if for every $W \subseteq \QVar$ and $(M, N) \in \dspec(W)$,
		\[
		\e \dpsat (M, N) \quad \text{implies} \quad \f \dpsat (M, N).
		\]
	\end{enumerate}
	
	Intuitively, $\f$ refines $\e$ if every specification satisfied by $\e$ is also satisfied by $\f$. This means $\f$ can safely replace $\e$ in any context where $\e$ was used.
	Let $\equiv_T^e$ denote the kernel of $\dtrefine$ (i.e., $\equiv_T^e {=} \dtrefine {\cap}\ {(\dtrefine)^{-1}}$), representing equivalence under total correctness refinement. Define $\equiv_P^e$ similarly for partial correctness.
	
	\subsubsection{Characterization Theorem}
	
	Our central result in this section establishes that effect-based refinement orders coincide precisely with the approximation ordering and its dual.
	
	\begin{theorem}[Effect-Based Refinement Characterization]\label{thm:det}
		For any $\e, \f \in \dprog(V)$,
		\begin{enumerate}
			\item $\e \dtrefine \f$ iff $\ \e \le \f$;
			\item $\e \dprefine \f$ iff $\ \e \ge \f$;
			\item $\e \equiv_T^e \f$ iff $\ \e \equiv_P^e \f$ iff $\ \e = \f$;
			\item If both $\e$ and $\f$ are trace-preserving, then $\e \dtrefine \f$ iff $\ \e \dprefine \f$ iff $\ \e = \f$.
		\end{enumerate}
	\end{theorem}
	
	\begin{proof}
		For (1), we have the following chain of equivalences:
		\begin{align*}
			\e \dtrefine \f
			&\quad\mbox{iff}\quad \forall W \subseteq \QVar, \forall M, N \in \p(\h_W): \e \dtsat (M, N) \Rightarrow \f \dtsat (M, N)
			\quad \text{(definition)}\\
			&\quad\mbox{iff}\quad \forall W \subseteq \QVar, \forall M, N \in \p(\h_W): M \le \e^\dagger(N) \Rightarrow M \le \f^\dagger(N)\quad \text{(as $wp.\e.N = \e^\dagger(N)$)}\\
			&\quad\mbox{iff}\quad \forall W \subseteq \QVar, \forall N \in \p(\h_W): \e^\dagger(N) \le \f^\dagger(N)
			\quad \text{(since $M$ is arbitrary)}\\
			&\quad\mbox{iff}\quad \e^\dagger \le \f^\dagger
			\quad \text{(definition of approximation order)}\\
			&\quad\mbox{iff}\quad \e \le \f
			\quad \text{(Lemma~\ref{lem:choi})}.
		\end{align*}
				The proof of (2) is similar, and clauses (3) and (4) then follow directly.
	\end{proof}
	
	This theorem reveals a fundamental duality: total correctness refinement follows the approximation order ($\e \le \f$), while partial correctness refinement follows the reverse order ($\e \ge \f$). This reflects our treatment of non-termination. For total correctness, non-termination is undesirable. So the zero program $0$ (which never terminates) is the minimum element under $\dtrefine$, meaning $0 \dtrefine \f$ for all $\f$. That is, any program refines the program that always diverges. In contrast, for partial correctness, non-termination is acceptable. So the zero program $0$ is the maximum element under $\dprefine$, meaning $\e \dprefine 0$ for all $\e$.
 	
	\subsubsection{Single Formula Characterization}
	
	As a consequence of Theorem~\ref{thm:det}, refinement can be verified using a single, canonical specification rather than checking all possible specifications.
	
	\begin{theorem}[Single Formula Characterization]\label{thm:dsingleformula}
		For any $\e, \f \in \dprog(V)$ and $V' \| V$,
		\begin{enumerate}
			\item $\e \dtrefine \f$ iff $\ \f \dtsat (wp.\e.\Omega_{V,V'}, \Omega_{V,V'})$;
			\item $\e \dprefine \f$ iff $\ \f \dpsat (wlp.\e.(I - \Omega)_{V,V'}, (I - \Omega)_{V,V'})$,
		\end{enumerate}
		where $\Omega_{V,V'}$ is the maximally entangled state on $\h_V \otimes \h_{V'}$.
	\end{theorem}
	
	\begin{proof}
		We prove (1); the proof of (2) is analogous.
		\begin{align*}
			\e \dtrefine \f
			&\quad\mbox{iff}\quad \e \le \f
			\quad \text{(Theorem~\ref{thm:det}(1))}\\
			&\quad\mbox{iff}\quad \e^\dagger(\Omega_{V,V'}) \le \f^\dagger(\Omega_{V,V'})
			\quad \text{(Lemma~\ref{lem:choi}(4) with $A = \Omega$)}\\
			&\quad\mbox{iff}\quad wp.\e.\Omega_{V,V'} \le wp.\f.\Omega_{V,V'}
			\quad \text{(since $wp.\e.\Omega = \e^\dagger(\Omega)$)}\\
			&\quad\mbox{iff}\quad \f \dtsat (wp.\e.\Omega_{V,V'}, \Omega_{V,V'})
			\quad \text{(definition of $\dtsat$)}. \qedhere
		\end{align*}
	\end{proof}
	
	This result provides a practical verification method: to check whether $\f$ refines $\e$, it suffices to verify a single specification determined by $\e$ and the maximally entangled state. This is particularly useful in automated verification tools.
	
	\subsection{Refinement Under Projector-Based Specifications}
	\label{subsec:detproj}
	
	 We have established a close connection between refinement orders and the complete positivity of super-operators when effects are used as state predicates in specifications. This naturally raises the question of whether refinement orders change when the specification language is weakened or strengthened.
	 This subsection addresses the first case by characterizing the refinement orders induced by projector-based specifications. The second case will be considered in the next section, where refinement orders for nondeterministic programs are studied.

	\subsubsection{Correctness Definitions}
	
	Since projectors are special effects, the satisfaction relations $\dtsat$ and $\dpsat$ defined in the previous subsection apply directly. However, as shown in~\cite{zhou2019applied}, for $\e \in \dprog(V)$ and projectors $P, Q \in \s(\h_W)$, these relations have equivalent qualitative formulations:
	
	\begin{itemize}
		\item $\e \atsat (P, Q)$ iff for any $\rho \in \d(\h_X)$ with $V \cup W \subseteq X$,
		\[
		\rho \models P \quad \text{implies} \quad \e(\rho) \models Q \text{ and } \tr(\e(\rho)) = \tr(\rho).
		\]
		
		\item $\e \apsat (P, Q)$ iff for any $\rho \in \d(\h_X)$ with $V \cup W \subseteq X$,
		\[
		\rho \models P \quad \text{implies} \quad \e(\rho) \models Q.
		\]
	\end{itemize}
		
	To work with projector-based refinement, we need predicate transformers that return projectors rather than general effects.
	
	\begin{definition}[Projector Predicate Transformers]\label{def:wlp}
		Given $\e \in \dprog(V)$ and $Q \in \s(\h_W)$,
		
		\begin{enumerate}
			\item the \emph{weakest precondition} $wp^p.\e.Q \in \s(\h_W)$ is the unique projector such that $\e \atsat (wp^p.\e.Q, Q)$, and
 			for any $P \in \s(\h_W)$ with $\e \atsat (P, Q)$, it holds $P \le wp^p.\e.Q$;
			
			\item the \emph{weakest liberal precondition} $wlp^p.\e.Q \in \s(\h_W)$ is similarly defined using $\apsat$ instead of $\atsat$;
			
			\item the \emph{strongest postcondition} $sp^p.\e.Q \in \s(\h_W)$ is the unique projector such that  $\e \apsat (Q, sp^p.\e.Q)$, and
			for any $R \in \s(\h_W)$ with $\e \apsat (Q, R)$, it holds $sp^p.\e.Q \le R$.
		\end{enumerate}
	\end{definition}
	
	Note that $wp^p.\e.Q$ differs from $wp.\e.Q $ defined in the previous subsection because the former is the weakest \emph{projector} precondition, while the latter is the weakest \emph{effect} precondition.
	
	Since $\s(\h_W)$ forms a complete lattice, these predicate transformers always exist. The following technical lemma provides the tools to compute them explicitly.
	
	\begin{lemma}\label{lem:EN}
		For any $A \in \p(\h)$ and $P \in \s(\h)$,
		\begin{enumerate}
			\item $E(A) = \mathcal{N}(I - A)$, where $E(A) = \{|\psi\rangle : A|\psi\rangle = |\psi\rangle\}$ is the eigenspace of $A$ for eigenvalue 1, and $\mathcal{N}(A)= \supp{A}^\bot$ denotes the null subspace of $A$;
			\item $P \le A$ iff $\ P \le E(A)$ iff $\ P \le \mathcal{N}(I - A)$.
		\end{enumerate}
	\end{lemma}
	
	\begin{proof}
		(1): For any $|\psi\rangle \in \h$,
		\begin{align*}
			|\psi\rangle \in E(A)
			&\quad\mbox{iff}\quad A|\psi\rangle = |\psi\rangle
			\quad \text{(definition of $E(A)$)}\\
			&\quad\mbox{iff}\quad (I - A)|\psi\rangle = 0\\
			&\quad\mbox{iff}\quad |\psi\rangle \in \mathcal{N}(I - A)
			\quad \text{(definition of null space)}.
		\end{align*}
		
		(2): Suppose $P \le A$. For any normalized $|\psi\rangle \in P$, we have $\langle\psi|P|\psi\rangle = 1$. From $P \le A$, we get $\langle\psi|A|\psi\rangle \geq 1$. But since $A \le I$, we also have $\langle\psi|A|\psi\rangle \leq 1$. Therefore $\langle\psi|A|\psi\rangle = 1$, which implies $A|\psi\rangle = |\psi\rangle$. Hence $|\psi\rangle \in E(A)$, proving $P \subseteq E(A)$, i.e., $P \le E(A)$. The converse part follows directly from the fact that $E(A)\le A$.
		
		Finally, the equivalence with $P\le \mathcal{N}(I - A)$ follows from part (1).
	\end{proof}
	
	Using this lemma, we can derive explicit formulas for the projector predicate transformers.
	
	\begin{lemma}[Explicit Formulas for Projector Transformers]\label{lem:wpwlpsp}
		For any $\e \in \dprog(V)$ and $Q \in \s(\h_W)$,
\[wp^p.\e.Q = E(\e^\dag(Q)),\quad wlp^p.\e.Q = \mathcal{N}(\e^\dag(Q^\bot)), \quad sp^p.\e.Q = \supp{\e(Q)}.\]
	\end{lemma}
	
	\begin{proof}
		For $wp^p$, we have
		\begin{align*}
			\e \atsat (P, Q)
			&\quad\mbox{iff}\quad P \le \e^\dagger(Q)
			\quad \text{(since $wp.\e.Q = \e^\dag(Q)$)}\\
			&\quad\mbox{iff}\quad P \le E(\e^\dagger(Q))
			\quad \text{(Lemma~\ref{lem:EN}(2))}.
		\end{align*}
		Thus $E(\e^\dagger(Q))$ is the largest projector $P$ satisfying $\e \atsat (P, Q)$.
		
		For $wlp^p$, 
		\begin{align*}
			\e \dpsat (P, Q)
			&\quad\mbox{iff}\quad P \le I - \e^\dagger(I - Q)
			\quad \text{(since $wlp.\e.Q = I - \e^\dagger(I - Q)$)}\\
			&\quad\mbox{iff}\quad P \le \mathcal{N}(\e^\dagger(I - Q))
			\quad \text{(Lemma~\ref{lem:EN}(2))}\\
			&\quad\mbox{iff}\quad P \le \mathcal{N}(\e^\dagger(Q^\bot))
			\quad \text{(since $I - Q = Q^\bot$ for projectors)}.
		\end{align*}
		Thus $\mathcal{N}(\e^\dagger(Q^\bot))$ is the largest projector $P$ satisfying $\e \dpsat (P, Q)$.
				
		For $sp^p$, assume $\dim(Q) > 0$ (the case $\dim(Q) = 0$ is trivial). Then
		\begin{align*}
			\e \apsat (Q, R)
			&\quad\mbox{iff}\quad \forall \rho: \rho \models Q \text{ implies } \e(\rho) \models R\\
			&\quad\mbox{iff}\quad \e(Q/\dim(Q)) \models R
			\quad  \text{(since $\supp{\e(\rho)} \le \supp{\e(Q/\dim(Q))}$ for all $\rho \models Q$)}\\
			&\quad\mbox{iff}\quad \supp{\e(Q)} \le R
			\quad \text{(since $\supp{\e(Q)} = \supp{\e(Q/\dim(Q))}$)}.
		\end{align*}
		Thus $\supp{\e(Q)}$ is the smallest projector $R$ satisfying $\e \apsat (Q, R)$.
	\end{proof}
	
		\subsubsection{Characterization of Projector-Based Refinement}
	
	With predicate transformers established, we can now characterize refinement for projector-based specifications. We write $\e \atrefine \f$ 
	if for every $W \subseteq \QVar$ and $P, Q \in \s(\h_W)$,
\[
\e \atsat (P, Q) \quad \text{implies} \quad \f \atsat (P, Q),
\]
	and define $\e \aprefine \f$ analogously for partial correctness.
	It is straightforward to show that 
	\begin{align}
		\e \atrefine \f &\quad \mbox{iff}\quad \forall Q: wp^p.\e.Q \le wp^p.\f.Q\label{eq:wpt}\\
		\e \aprefine \f & \quad \mbox{iff}\quad  \forall Q: wlp^p.\e.Q \le wlp^p.\f.Q \quad \mbox{iff}\quad \forall P: sp^p.\f.P \le sp^p.\e.P. \label{eq:wlpp}
	\end{align}
 Furthermore, analogous to Theorem~\ref{thm:det}, we can provide a more direct characterization of the refinement orders $\aprefine$ and $\atrefine$. To this end, we introduce two key notions. For a super-operator $\e$ with Kraus operators $\{E_i : i \in K\}$, define
	\[
	\spann(\e) = \spann\{E_i : i \in K\},
	\]
	the linear span of the Kraus operators. This is well-defined as the span is independent of the choice of Kraus representation. Specifically, if $\{E_i: i\in K\}$ and $\{E_i': i\in K'\}$ both correspond to $\e$, then $\spann\{E_i: i\in K\}=\spann\{E_i': i\in K'\}$. 
	
	The \emph{termination space} of $\e \in \dprog(V)$ is defined as~\cite{zhou2019applied}
	\[
	T_\e = \left\{|\psi\rangle \in \h_V : \tr(\e(|\psi\rangle\langle\psi|)) = \tr(|\psi\rangle\langle\psi|)\right\},
	\]
	the subspace of states on which $\e$ preserves probability (terminates). With our notations, $T_\e = wp^p.\e.I_V = E(\e^\dagger(I_V))$.

	\begin{theorem}[Projector-Based Refinement Characterization]\label{thm:detproj}
		For any $\e, \f \in \dprog(V)$,
		\begin{enumerate}
			\item $\e \aprefine \f$ iff $\ \spann(\f) \le \spann(\e)$;
			\item $\e \atrefine \f$ iff $\ T_\e \le T_\f$ and $\spann(\f \circ \p_{T_\e}) \le \spann(\e \circ \p_{T_\e})$, where $\p_{T_\e}$ is the super-operator with $T_\e$ as its unique Kraus operator;
			\item If both $\e$ and $\f$ are trace-preserving, then $\e \atrefine \f$ iff $\ \e \aprefine \f$.
		\end{enumerate}
	\end{theorem}
	
	\begin{proof}
		For (1), we establish the following chain of equivalences:
		\begin{align*}
			\e \aprefine \f
			&\quad\mbox{iff}\quad \forall W\subseteq \QVar, \forall P \in \s(\h_W): \supp{\f(P)} \le \supp{\e(P)}
			\quad \text{(by Lemma~\ref{lem:wpwlpsp} and Eq.~\eqref{eq:wlpp})}\\
			&\quad\mbox{iff}\quad \forall W\subseteq \QVar, \forall |\psi\rangle \in \h_W: \spann\{F_j|\psi\rangle : j \in J\} \le \spann\{E_i|\psi\rangle : i \in K\}\\
			&\qquad\qquad \text{(where $\{E_i\}$, $\{F_j\}$ are Kraus operators of $\e$, $\f$)}\\
			&\quad\mbox{iff}\quad \spann\{F_j : j \in J\} \le \spann\{E_i : i \in K\} \quad \text{(see below)}.
		\end{align*}
To justify the last equivalence, the `if' direction is straightforward. For the `only if' direction, assume that $\spann\{F_j|\psi\>: j\in J\}\le \spann\{E_i|\psi\>: i\in K\}$ for any $|\psi\>$. We need to prove $\spann\{F_j: j\in J\}\le \spann\{E_i: i\in K\}$. Otherwise, there exist some $j_0\in J$ and $|x\>,|y\>\in \h_V$ such that $\<x|E_i|y\>=0$ for all $i\in K$, yet $\<x|F_{j_0}|y\>\neq 0$. However, from the assumption, we have
$
F_{j_0}|y\> = \sum_{i \in K} \lambda_{i,j_0}E_i |y\>
$
for some coefficients $\lambda_{i,j_0}$. Thus 
\[
\<x|F_{j_0}|y\> = \sum_{i \in K} \lambda_{i,j_0}\<x|E_i |y\>=0,
\]
a contradiction, and so \( \spann\{F_j: j \in J\} \le \spann\{E_i: i \in K\} \) indeed holds.

For (2), we have on one hand,
\begin{align}
	\e\atrefine \f &\quad \mbox{iff}\quad\forall W\subseteq \QVar,\  \forall Q\in \s(\h_W): wp^p.\e.Q\le wp^p.\f.Q \quad \text{(by Eq.~\eqref{eq:wpt})} \notag \\
	&\quad \mbox{iff}\quad\forall W\subseteq \QVar,\  \forall Q\in \s(\h_W): T_\e \wedge wlp^p.\e.Q\le T_\f \wedge wlp^p.\f.Q \label{eq:tmp1}\\
	&\qquad\qquad\qquad \text{(by Theorem~\ref{thm:galois}(2))}. \notag
\end{align} 
On the other hand, for all $W\subseteq \QVar$ and $Q\in \s(\h_W)$,
\begin{align}
	&\quad wlp^p.\left(\e\circ \p_{T_\e}\right).Q\le wlp^p.\left(\f\circ \p_{T_\e}\right).Q\notag\\
	\mbox{iff}&\quad T_\e^\bot \vee (T_\e \wedge wlp^p.\e.Q)\le T_\e^\bot \vee (T_\e \wedge wlp^p.\f.Q)\notag\\
	& \qquad\qquad \text{(Lemma 4.2 of \cite{feng2023refinement})}\notag\\
	\mbox{iff}&\quad T_\e \wedge wlp^p.\e.Q\le T_\e \wedge wlp^p.\f.Q \quad \text{(since $T_\e\ {\bot}\ T_\e^\bot$)} \label{eq:tmp2}.
\end{align} 

We now show Eq.~\eqref{eq:tmp1} is equivalent to the conjunction of Eq.~\eqref{eq:tmp2} and $T_\e \le T_\f$. First, from Eq.~\eqref{eq:tmp1}, we have $T_\e \le T_\f$ by taking $Q=I_W$ and noting Theorem~\ref{thm:ptransformdet}(1). Furthermore, by taking the meet of $T_\e$ with both sides of Eq.~\eqref{eq:tmp1}, we derive Eq.~\eqref{eq:tmp2}.
Conversely, suppose $T_\e \le T_\f$. Then from Eq.~\eqref{eq:tmp2} we immediately have Eq.~\eqref{eq:tmp1}.

Thus we have $\e\atrefine \f$ iff $T_\e \le T_\f$ and $\e\circ \p_{T_\e}\aprefine \f\circ \p_{T_\e}$. Then (2) follows from (1). Finally, clause (3) follows directly from the fact that total correctness and partial correctness are equivalent for trace-preserving quantum programs. It can also be easily seen from clauses (1) and (2), since when both $\e$ and $\f$ are trace-preserving, $T_\e = T_\f = I_V$, and $\p_{T_\e} = \id_{V}$ is the identity super-operator.
	\end{proof}
	
	\subsubsection{Comparison: Effects vs. Projectors}
	
	We now establish that projector-based refinement is strictly weaker than effect-based refinement.
	
	\begin{proposition}\label{lem:eimpp}
		For deterministic programs $\e, \f \in \dprog(V)$,
		\begin{enumerate}
			\item $\e \dtrefine \f$ implies $\e \atrefine \f$, but the converse is false;
			\item $\e \dprefine \f$ implies $\e \aprefine \f$, but the converse is false.
		\end{enumerate}
	\end{proposition}
	
	\begin{proof}
		Direct from Theorems~\ref{thm:det} and~\ref{thm:detproj}. To be more explicit, we present counterexamples for the reverse parts of both clauses. Consider $\e = \frac{1}{2}\id$ and $\f = \frac{1}{3}\id$, where $\id$ is the identity super-operator on $\h_V$. Then it can be easily checked that both $\e\equiv_T^p \f$ and $\e\equiv_P^p \f$ but neither $\e \dtrefine \f$ nor  $\f \dprefine \e$. Actually, if ${\e} = p\id$ with $0<p<1$, then $\e\equiv^p_T 0$ and $\e\equiv^p_P \id$.
	\end{proof}
	
	\subsection{Properties of Projector Predicate Transformers}
	
	For completeness, we collect several important properties of the projector predicate transformers that may be useful in practical reasoning about quantum programs.
	
	\begin{theorem}\label{thm:galois}
		For any $\e \in \dprog(V)$ and projectors $P, Q \in \s(\h_W)$,
		\begin{enumerate}
			\item $sp^p.\e.(wlp^p.\e.Q) \le Q$ and $P \le wlp^p.\e.(sp^p.\e.P)$;
			\item $wp^p.\e.Q = wlp^p.\e.Q \wedge wp^p.\e.I_W$;
			\item $sp^p.\e.(wp^p.\e.Q) \le Q$, and if $P \le wp^p.\e.I_W$, then $P \le wp^p.\e.(sp^p.\e.P)$.
		\end{enumerate}
	\end{theorem}
	
	\begin{proof}
		 (1) follows from the fact that $wlp^p.\e$ and $sp^p.\e$: $P \le wlp^p.\e.Q$ iff $\ sp^p.\e.P \le Q$. For part (2), we need to show $E(\e^\dagger(Q)) = \mathcal{N}(\e^\dagger(Q^\bot)) \wedge E(\e^\dagger(I_W))$. For any normalized state $|\psi\rangle \in \h_V$,
		\begin{align*}
			|\psi\rangle \in E(\e^\dagger(Q))
			&\quad\mbox{iff}\quad \langle\psi|\e^\dagger(Q)|\psi\rangle = 1 \qquad \qquad \text{(definition of $E(\cdot)$)}\\
			&\quad\mbox{iff}\quad \tr(Q\e(\psi)) = 1 \qquad \qquad \text{(definition of $\e^\dag$, where $\psi=|\psi\rangle\langle\psi|$)}\\
			&\quad\mbox{iff}\quad \tr(\e(\psi)) = 1 \text{ and } \tr(Q\e(\psi)) = \tr(\e(\psi))\\
			&\hspace{14em} \text{(since $\tr(Q\e(\psi)) \leq \tr(\e(\psi)) \leq 1$)}\\
			&\quad\mbox{iff}\quad \langle\psi|\e^\dagger(I_W)|\psi\rangle = 1  \text{ and } \tr(Q^\bot\e(\psi)) = 0\qquad \text{(definitions of $\e^\dag$ and $Q^\bot$)}\\
			&\quad\mbox{iff}\quad |\psi\rangle \in E(\e^\dagger(I_W)) \text{ and } |\psi\rangle \in \mathcal{N}(\e^\dagger(Q^\bot))\quad  \text{(definitions of $E(\cdot)$ and $\mathcal{N}(\cdot)$)}\\
			&\quad\mbox{iff}\quad |\psi\rangle \in E(\e^\dagger(I_W)) \wedge \mathcal{N}(\e^\dagger(Q^\bot)) \qquad \text{(definitions of $\wedge$)}.
		\end{align*}
		The first part of (3) follows from (2) and (1). For the second part, we first know from (1) that $P\le wlp^p.\e.(sp^p.\e.P)$. Then $P\le wp^p.\e.(sp^p.\e.P)$ follows from the assumption that $P\le wp^p.\e.I_W$ and (2) by taking $Q = sp^p.\e.P$.
	\end{proof}
	
	The pair $(wlp^p.\e, sp^p.\e)$ thus forms a Galois connection between the lattice $\s(\h_W)$ and itself, providing an adjunction relationship between forward and backward reasoning about quantum programs.
	
	\begin{theorem}\label{thm:ptransformdet}
		For any $\e \in \dprog(V)$, projectors $P, Q \in \s(\h_W)$, and $xp^p \in \{wp^p, wlp^p\}$,
		\begin{enumerate}
			\item $wp^p.\e.0 = sp^p.\e.0 = 0$ and $wlp^p.\e.I_W = I_W$;
			\item (Monotonicity) If $P \le Q$, then $xp^p.\e.P \le xp^p.\e.Q$ and $sp^p.\e.P \le sp^p.\e.Q$;
			\item  $xp^p.\e.(P \wedge Q) = xp^p.\e.P \wedge xp^p.\e.Q$;
			\item  $sp^p.\e.(P \wedge Q) \le sp^p.\e.P \wedge sp^p.\e.Q$;
			\item $xp^p.\e.P \vee xp^p.\e.Q \le xp^p.\e.(P \vee Q)$;
			\item $sp^p.\e.P \vee sp^p.\e.Q = sp^p.\e.(P \vee Q)$;
			\item (Frame axiom) If $R \in \s(\h_{W'})$ with $W' \cap (V \cup W) = \emptyset$, then
			\begin{align*}
				wp^p.\e.(Q \otimes R) &= (wp^p.\e.Q) \otimes R,\\
				sp^p.\e.(P \otimes R) &= (sp^p.\e.P) \otimes R,\\
				wlp^p.\e.(Q \otimes R) &\ge (wlp^p.\e.Q) \otimes R.
			\end{align*}
		\end{enumerate}
	\end{theorem}
		\begin{proof}
				(1) is easy to check; (2), (4), and (5) are from the fact that $E(\cdot)$ is monotonic with respect to $\le$. For (3),
	it suffices to prove that
	\[
	E(\e^\dag(P)) \wedge E(\e^\dag(Q)) \le E(\e^\dag(P\wedge Q)).
	\]
	Given any normalized state $|\psi\>$ in the left-hand side subspace, we have
	\[
	\<\psi|\e^\dag(P)|\psi\> = \<\psi|\e^\dag(Q)|\psi\> = 1,
	\]
	which implies $
	\tr (P\e(\psi))= 	\tr (Q\e(\psi))=1$ where $\psi = |\psi\>\<\psi|$. Thus the state $\e(\psi)$ lies in both $P$ and $Q$ and $\tr (\e(\psi)) = 1$. Consequently, 
	$\tr ((P\wedge Q)\e(\psi)) =1$, and so $|\psi\>$ is in the right-hand side subspace as well.
	
	For (6),
	it suffices to prove that
	\[
	\supp{\e(P\vee Q)} \le \supp{\e(P)} \vee \supp{\e(Q)}.
	\]
	Given any normalized state $|\psi\>$ orthogonal to the right-hand side subspace, we have
	$
	\<\psi|\e(P)|\psi\> = \<\psi|\e(Q)|\psi\> = 0,
	$
	which implies $
	P\le \mathcal{N}(\e^\dag(\psi))$ and $
	Q\le \mathcal{N}(\e^\dag(\psi))$. Thus $
	P\vee Q\le \mathcal{N}(\e^\dag(\psi))$, or equivalently, $\tr((P\vee Q) \e^\dag(\psi)) = 0$. Hence $
	\<\psi|\e(P\vee Q)|\psi\>  = 0
	$, and so
	$|\psi\>$ is orthogonal to the left-hand side subspace as well.
	
	For (7), note that when $W'\cap (V\cup W) = \emptyset$, 
	\begin{align*}
		E(\e^\dag(P\otimes R)) &= E(\e^\dag(P)\otimes R) = E(\e^\dag(P))\otimes R\\
		\supp{\e(P\otimes R)} &= \supp{\e(P)\otimes R} = \supp{\e(P)}\otimes R.
	\end{align*}
	Thus the first two equalities hold trivially. For the last one, it suffices to prove that $$\mathcal{N}\left(\e^\dag\left[(Q\otimes R)^\bot\right]\right)\ge \mathcal{N}\left(\e^\dag(Q^\bot)\right)\otimes R.$$
	For any $|\phi\>\in R$ and $|\psi\>$ in the null space of $\e^\dag(Q^\bot)$, we compute that
	\begin{align*}
		\tr\left(\e^\dag\left[(Q\otimes R)^\bot\right] \left(\psi\otimes \phi\right)\right) & = \tr\left((Q\otimes R)^\bot \left[\e(\psi)\otimes \phi\right]\right)\qquad \text{(definition of $\e^\dag$)}\\
		& = \tr\left(\left(Q^\bot + Q\otimes R^\bot\right) \left[\e(\psi)\otimes \phi\right]\right) \ \  \text{(since $(Q\otimes R)^\bot=\left(Q^\bot + Q\otimes R^\bot\right)$)}\\
		& = \tr\left(\e^\dag(Q^\bot) \psi\right)+ \tr(Q\e(\psi))\tr \left(R^\bot \phi\right)\quad \text{(properties of $\tr(\cdot)$)}\\
		&=0 \qquad \text{(assumption that $|\psi\>\in \mathcal{N}(\e^\dag(Q^\bot))$ and $|\phi\>\in R$)}.
	\end{align*}
	That completes the proof.
	\end{proof}
	
	Finally, we note that the inequalities in clauses (4), (5), and (7) of the above theorem generally hold strictly.
	
	For (4), let \( \e = \text{Set}_q^0 \), which sets the target qubit \( q \) to the state \( |0\rangle \). Specifically, \( \text{Set}_q^0(\rho) = |0\rangle_q \langle 0|\rho |0\rangle_q \langle 0| + |1\rangle_q \langle 1|\rho |0\rangle_q \langle 0| \). Let \( P = |0\rangle_q \langle 0| \) and \( Q = |+\rangle_q \langle +| \). Then, \( sp^p.\e.(P \wedge Q) = sp^p.\e.0 = 0 \), but \( sp^p.\e.P = sp^p.\e.Q = |0\rangle_q \langle 0| \).
	
	For (5), let \( \e = \text{Set}_q^0 \), \( P = |+\rangle_q \langle +| \), and \( Q = |-\rangle_q \langle -| \). Thus, \( wp^p.\e.P = wp^p.\e.Q = 0 \). However, \( wp^p.\e.(P \vee Q) = wp^p.\e.I = I \). Since \( \e \) is trace-preserving, the same example applies to \( wlp^p.\e \).
	
	For (7), let \( \e = 0 \), the zero super-operator, \( R = 0 \), the zero operator, and \( P \) be arbitrary. Then, \( wlp^p.\e.(P \otimes R) = wlp^p.\e.0 = I \), but \( (wlp^p.\e.P) \otimes R = 0 \).

	\section{Refinement for Nondeterministic Quantum Programs}
	\label{sec:nondeterministic}
	
	Nondeterminism is a fundamental feature of realistic programs, arising from underspecification, abstraction, or genuine randomness in program behavior. It is also useful in the development of refinement calculus, since the semantics of a prescription is typically described by the set of executable programs that satisfy the prescription~\cite{peduri2025qbc,feng2023refinement}. This section extends our refinement analysis from deterministic to nondeterministic quantum programs. We establish that refinement orders for nondeterministic programs, when based on sets of effects, correspond precisely to the classical Hoare and Smyth orders from domain theory. We then investigate how these refinement orders degrade when we restrict to simpler predicate classes (individual effects or projectors), showing that each restriction yields strictly weaker refinement notions.
	
	\subsection{Semantic Model for Nondeterminism}
	
	Following established practice in program semantics~\cite{Morgan:1996,jifeng1997probabilistic,mciver2001partial}, we model nondeterministic quantum programs as sets of deterministic programs. However, not all sets are appropriate; we require additional structure to ensure mathematical well-behavedness. To be specific, for a finite set $V \subseteq \QVar$ of quantum variables, we define
	\[
	\nprog(V) = \{\E \subseteq \dprog(V) : \E \text{ is nonempty, convex, and closed}\}
	\]
	as the semantic space of nondeterministic quantum programs on system $V$.
	Each condition serves a specific purpose:
	\begin{itemize}
		\item \emph{Nonempty}. We exclude $\emptyset$ to avoid degeneracies in our definitions and keep the presentation clean. An empty set of behaviors has no clear operational meaning.
		
		\item \emph{Convex}. If super-operators $\e$ and $\f$ are both possible behaviors, then any probabilistic mixture $p\e + (1-p)\f$ (for $0 \leq p \leq 1$) should also be a possible behavior. Convexity captures the idea that probabilistic composition of nondeterministic choices remains a valid nondeterministic choice. This is standard in probabilistic program semantics~\cite{Morgan:1996,jifeng1997probabilistic,mciver2001partial}.
		
		\item \emph{Closed}. Since $\dprog(V)$ is finite-dimensional as $V$ is a finite set, closedness here is equivalent to being closed in the Euclidean topology. This requirement ensures that limits of sequences in $\E$ remain in $\E$, providing mathematical robustness for analysis and approximation arguments.
	\end{itemize}
	
	Note that unlike some classical approaches, we do \emph{not} require up-closedness, as we wish to distinguish between total and partial correctness perspectives.
	
	\subsection{Refinement Under Set-of-Effects Specifications}
	
	We begin with the richest class of predicates—sets of effects—which, as we shall demonstrate, yields the cleanest characterization of refinement for nondeterministic programs.
	Following~\cite{feng2023verification}, for a finite set $W \subseteq \QVar$, define
	\[
	\nspec(W) = \left\{(\qassert, \qassertp) : \qassert, \qassertp \in 2^{\p(\h_W)}\right\}
	\]
	as the set of specifications for nondeterministic quantum programs on system $W$. Here $\qassert$ represents the precondition and $\qassertp$ the postcondition, both as sets of effects.
	
	\subsubsection{Correctness Definitions and Refinement Orders}
	
	 Similarly to the deterministic case, there are two different notions of satisfaction of a nondeterministic quantum program $\E$ in $\nprog(V)$ on a specification $(\qassert,\qassertp)$ in $\nspec(W)$, one for total correctness, and the other for partial correctness. We say  
	\begin{enumerate}
		\item $\E$ satisfies $(\qassert,\qassertp)$ in the sense of total correctness, denoted $\E\ntsat (\qassert,\qassertp)$, if for any $\rho\in \d(\h_X)$ with $V\cup W\subseteq X$,
		$$	\Exp_{\mathsf{dem}}(\qstate \models \qassert)\leq \inf \left\{	\Exp_{\mathsf{dem}}(\sigma \models \qassertp) : \sigma\in \E(\qstate) \right\}.$$
			This satisfaction relation can be equivalently expressed using predicate transformers. Extending the adjoint operation pointwise to sets of super-operators
		\[
		\E^\dagger(\qassertp) = \left\{\e^\dagger(N) : \e \in \E, N \in \qassertp\right\},
		\]
		we have $\E \ntsat (\qassert, \qassertp)$ iff $\ \qassert \leinf \E^\dagger(\qassertp)$. Thus, the \emph{weakest precondition} $wp^s.\E.\qassertp$ of $\E$ with respect to $\qassertp$ exists and equals $\E^\dagger(\qassertp).
		$
		\item $\E$ satisfies $(\qassert,\qassertp)$ in the sense of partial correctness, denoted $\E\npsat (\qassert,\qassertp)$, if for any $\rho\in \d(\h_X)$ with $V\cup W\subseteq X$,
		$$		\Exp_{\mathsf{dem}}(\qstate \models \qassert)\leq \inf \left\{	\Exp_{\mathsf{dem}}(\sigma \models \qassertp) + \tr(\rho) - \tr(\sigma): \sigma\in \E(\qstate) \right\}.$$
			The additional term $\tr(\rho) - \tr(\sigma)$ accounts for non-termination (trace loss), as in the deterministic case. Equivalently, $\E \npsat (\qassert, \qassertp)$ iff $\qassert \leinf I - \E^\dagger(I - \qassertp)$ iff $\E^\dagger(I - \qassertp) \lesup I - \qassert$. Thus the weakest liberal precondition $wlp^s.\E.\qassertp$ of $\E$ with respect to $\qassertp$ exists and equals $ I-\E^\dag(I-\qassertp)$.
	\end{enumerate}
	For the purpose of refinement, nondeterminism in a program is interpreted in a demonic way. Thus we use $\Exp_{\mathsf{dem}}$ instead of $\Exp_{\mathsf{ang}}$ in both total and partial correctness definitions.
	
	As in the deterministic case, we define refinement based on preservation of satisfied specifications. To be specific, we write $\E\ntrefine \F$ if for any $(\qassert,\qassertp)\in \nspec(W)$, $\E\ntsat (\qassert,\qassertp)$ implies $\F\ntsat (\qassert,\qassertp)$. In contrast, $\E\nprefine \F$ if for any $\qassert$ and $\qassertp$, $\E\npsat (\qassert,\qassertp)$ implies $\F\npsat (\qassert,\qassertp)$. 
	
	\subsubsection{Main Characterization: Connection to Domain Theory}
	
	Our central result for nondeterministic programs establishes a precise correspondence between refinement orders and the Hoare and Smyth orders from classical domain theory. 
	
	\begin{theorem}[Set-of-Effects Refinement Characterization]\label{thm:nondet}
		For any $\E, \F \in \nprog(V)$,
		\begin{enumerate}
			\item $\E \ntrefine \F$ iff $\ \E \leq_S \F$ iff $\upcl \F \subseteq \upcl \E$;
			\item $\E \nprefine \F$ iff $\ \F \leq_H \E$ iff $\downcl \F \subseteq \downcl \E$;
			\item $\E \leq_{EM} \F$ iff both $\E \ntrefine \F$ and $\F \nprefine \E$;
			\item $\E \equiv_T^s \F$ iff $\upcl \E = \upcl \F$, and $\E \equiv_P^s \F$ iff $\downcl \E = \downcl \F$;
			\item If $\ \E$ and $\F$ contain only trace-preserving super-operators, then 
			\[
			\E \ntrefine \F\quad \mbox{iff}\quad \E \nprefine \F \quad \mbox{iff}\quad \F \subseteq \E.
			\]
		\end{enumerate}
	\end{theorem}
	
	\begin{proof}
		Let $\mathcal{A} = \bigcup_{W \subseteq \QVar} \p(\h_W)$ denote the collection of all effects over all possible finite subsystems.
		
		For (1), we establish a chain of equivalences connecting refinement to the Smyth order:
		\begin{align*}
			\E \ntrefine \F
			&\quad\mbox{iff}\quad \forall \qassert, \qassertp \in \mathcal{A}: \qassert \leinf \E^\dagger(\qassertp) \Rightarrow \qassert \leinf \F^\dagger(\qassertp)
			\quad \text{(since $wp^s.\E.\qassertp = \E^\dag(\qassertp)$)}\\
			&\quad\mbox{iff}\quad \forall \qassertp \subseteq \mathcal{A}: \E^\dagger(\qassertp) \leinf \F^\dagger(\qassertp)
			\quad \text{(taking $\qassert= \E^\dag(\qassertp)$)}\\
			&\quad\mbox{iff}\quad \forall \qassertp \subseteq \mathcal{A}: \E^\dagger(\qassertp) \leq_S \F^\dagger(\qassertp)
			\quad \text{(Theorem~\ref{thm:equivalence})}\\
			&\quad\mbox{iff}\quad \forall N \in \mathcal{A}: \E^\dagger(N) \leq_S \F^\dagger(N)
			\quad \text{(taking $N\in \qassertp$ arbitrarily)}\\
			&\quad\mbox{iff}\quad \forall N \in \mathcal{A}, \forall \f \in \F, \exists \e \in \E: \e^\dagger(N) \le \f^\dagger(N)
			\quad \text{(definition of $\leq_S$)}\\
			&\quad\mbox{iff}\quad \forall \f \in \F, \exists \e \in \E, \forall N \in \mathcal{A}: \e^\dagger(N) \le \f^\dagger(N)
			\quad \text{(key step; see below)}\\
			&\quad\mbox{iff}\quad \forall \f \in \F, \exists \e \in \E: \e \le \f
			\quad \text{(Lemma~\ref{lem:choi})}\\
			&\quad\mbox{iff}\quad \E \leq_S \F
			\quad \text{(definition of Smyth order)}.
		\end{align*}

		The key step above requires justification. The `if' part is immediate. For the reverse, we use the minimax theorem (Theorem~\ref{thm:sm}). Since $\E$ and $\F$ are convex and closed, and the sets $\d(\h_{V \cup V'})$, $\E^\dagger(N)$, and $\F^\dagger(N)$ are all convex and compact, we can apply the Sion minimax theorem exactly as in Theorem~\ref{thm:equivalence} to exchange the order of quantifiers over predicates and program behaviors.
		
		The equivalence with $\upcl \F \subseteq \upcl \E$ follows from the standard characterization of the Smyth order via up-closures (Section~\ref{sec:preliminaries}).
		
		The argument of (2) is similar. Moreover, (3) and (4) are direct from (1) and (2). Finally, 
%
		for (5), when all super-operators in $\E$ and $\F$ are trace-preserving, non-termination is impossible, so total and partial correctness coincide. Moreover, for trace-preserving programs, the approximation order $\e \le \f$ becomes equality $\e = \f$. Thus, both $\E \leq_S \F$ and $\F \leq_H \E$ reduce to $\F \subseteq \E$.
	\end{proof}
	
	Finally, we prove that using effects or sets of effects as quantum state predicates yields the same refinement orders for deterministic quantum programs. This resolves one of the questions posed in Subsection~\ref{subsec:detproj}.
	
	\begin{corollary}\label{cor:nondet}
		For deterministic programs $\e, \f \in \dprog(V)$,
		\begin{enumerate}
			\item $\{\e\} \ntrefine \{\f\}$ iff $\ \e \dtrefine \f$;
			\item $\{\e\} \nprefine \{\f\}$ iff $\ \e \dprefine \f$.
		\end{enumerate}
	\end{corollary}
	
	\begin{proof}
		For singleton sets, the Smyth and Hoare orders both reduce to the underlying approximation order: $\{\e\} \leq_S \{\f\}$ iff $\{\e\} \leq_H \{\f\}$ iff $\ \e \le \f$. The result then follows immediately from Theorems~\ref{thm:det} and~\ref{thm:nondet}.
	\end{proof}
	
	\subsubsection{Single-Formula Characterization}
	
	As in the deterministic case, refinement can be verified using a single canonical specification.
	
	\begin{theorem}[Single Formula for Nondeterministic Programs]\label{thm:nsingleformula}
		For any $\E, \F \in \nprog(V)$ and $V' \| V$,
		\begin{enumerate}
			\item $\E \ntrefine \F$ iff $\ \F \ntsat (wp^s.\E.\Omega_{V,V'}, \Omega_{V,V'})$;
			\item $\E \nprefine \F$ iff $\ \F \npsat (wlp^s.\E.(I - \Omega_{V,V'}), I - \Omega_{V,V'})$.
		\end{enumerate}
	\end{theorem}
	
	\begin{proof}
		We prove (2); the proof of (1) is analogous.
		\begin{align*}
			\E \nprefine \F
			&\quad\mbox{iff}\quad \F \leq_H \E
			\quad \text{(Theorem~\ref{thm:nondet}(2))}\\
			&\quad\mbox{iff}\quad \forall \e \in \E, \exists \f \in \F: \f \le \e
			\quad \text{(definition of $\leq_H$)}\\
			&\quad\mbox{iff}\quad \forall \e \in \E, \exists \f \in \F: \f^\dagger(\Omega_{V,V'}) \le \e^\dagger(\Omega_{V,V'})
			\quad \text{(Lemma~\ref{lem:choi})}\\
			&\quad\mbox{iff}\quad \F^\dagger(\Omega_{V,V'}) \leq_H \E^\dagger(\Omega_{V,V'})
			\quad \text{(definition of $\leq_H$ for sets)}\\
			&\quad\mbox{iff}\quad \F^\dagger(\Omega_{V,V'}) \lesup \E^\dagger(\Omega_{V,V'})
			\quad \text{(Theorem~\ref{thm:equivalence}(2))}\\
			&\quad\mbox{iff}\quad I - \E^\dagger(\Omega_{V,V'}) \leinf I - \F^\dagger(\Omega_{V,V'})
			\quad \text{(Lemma~\ref{lem:duality}(1))}\\
			&\quad\mbox{iff}\quad \F \npsat (I - \E^\dagger(\Omega_{V,V'}), I - \Omega_{V,V'})
			\quad \text{(definition of $\npsat$)}\\
			&\quad\mbox{iff}\quad \F \npsat (wlp^s.\E.(I - \Omega_{V,V'}), I - \Omega_{V,V'})
			\quad \text{(definition of $wlp^s$)}. \qedhere
		\end{align*}
	\end{proof}
	
	Again, Theorem~\ref{thm:nsingleformula} provides a practical verification method: to check whether $\F$ refines $\E$, it suffices to verify a single specification determined by $\E$ and the maximally entangled state.
	
	\subsection{Refinement Under Effect-Based Specifications}
	
	Having characterized refinement using the most expressive predicates (sets of effects), we now investigate whether restricting to simpler predicates alters the refinement orders. This subsection examines the case of individual effects.
	
	We define refinement for nondeterministic programs under effect-based specifications by treating individual effects as singleton sets of effects. Specifically, $\E \dtrefine \F$ if for every $(M, N) \in \dspec(W)$,
	\[
	\E \dtsat (\{M\}, \{N\}) \quad \text{implies} \quad \F \dtsat (\{M\}, \{N\}).
	\]
	The relation $\E \dprefine \F$ is defined analogously for partial correctness.
		
	While we are not able to provide a complete characterization of $\dtrefine$ and $\dprefine$  as clean as Theorem~\ref{thm:nondet}, we can establish that effect-based refinement is strictly weaker than set-of-effects based refinement.
	
	\begin{proposition}[Effects Weaker Than Effect Sets]\label{prop:simpe}
		For nondeterministic programs $\E, \F \in \nprog(V)$,
		\begin{enumerate}
			\item $\E \ntrefine \F$ implies $\E \dtrefine \F$, but the converse is false;
			\item $\E \nprefine \F$ implies $\E \dprefine \F$, but the converse is false.
		\end{enumerate}
	\end{proposition}
	
	\begin{proof}
		The implications are straightforward.
		For the reverse part of (1), let $\E = \text{conv}\{\e_0, \e_1\}$ and $\F = \{\e_+\}$, where $\text{conv}$ denotes convex hull, and each $\e_x$ has a single Kraus operator $|x\rangle\langle x|$ for $x \in \{0, 1, +\}$. That is,
		\[
		\e_x(\rho) = |x\rangle\langle x|\rho|x\rangle\langle x| \quad \text{for all } \rho.
		\]
		Suppose $\E \dtsat (\{M\}, \{N\})$ for some effects $M, N$. Then for any state $\rho$,
		\[
		\tr(M\rho) \leq \min \left\{p \cdot \tr(N\e_0(\rho)) + (1-p) \cdot \tr(N\e_1(\rho)) : 0 \leq p \leq 1\right\}.
		\]
		
		Taking $\rho = |0\rangle\langle 0|$ and $|1\rangle\langle 1|$ respectively, we get
	\begin{align*}
	\<0|M|0\> & \leq \min\left\{p\cdot \<0|N|0\>: 0\leq p\leq 1\right\} = 0,\\
	\<1|M|1\> & \leq \min\left\{(1-p)\cdot \<1|N|1\>: 0\leq p\leq 1\right\} = 0.		
\end{align*}
		Thus $\tr(M) = \langle 0|M|0\rangle + \langle 1|M|1\rangle = 0$, implying $M = 0$ since $M$ is positive. As $\F \dtsat (\{0\}, \{N\})$ trivially holds for any $N$, we conclude $\E \dtrefine \F$.
		
		However, $\E \not\ntrefine \F$ because $\E \not\leq_S \F$. To see this, note that $\e_+ \in \F$, but for no convex combination $p\e_0 + (1-p)\e_1$ do we have $p\e_0 + (1-p)\e_1 \le \e_+$. Indeed, the Kraus operator $|+\rangle\langle +|$ is not in the span of $\{|0\rangle\langle 0|, |1\rangle\langle 1|\}$, so no convex combination can approximate it in the approximation order.
		
		For the reverse part of (2), let $\E = \text{conv}\{\e_\psi : |\psi\rangle \in \h_q, \<\psi|\psi\> =1\}$ (all rank-one projections) and $\F = \{\id\}$ (identity), where each $\e_\psi$ has Kraus operator $|\psi\rangle\langle\psi|$.
		
		Suppose $\E \dpsat (\{M\}, \{N\})$ for effects $M, N$. Then for any $\rho$ and any $|\psi\rangle$, 
		\[
		\tr((I - M)\rho) \geq \tr(|\psi\rangle\langle\psi|(I - N)|\psi\rangle\langle\psi|\rho).
		\]
		Taking $\rho = |\psi\rangle\langle\psi|$, we have
		\[
		\langle\psi|(I - M)|\psi\rangle \geq \langle\psi|(I - N)|\psi\rangle.
		\]
		Since this holds for all $|\psi\rangle$, we must have $I - M \ge I - N$, i.e., $M \le N$. But this means $\F \dpsat (\{M\}, \{N\})$, so $\E \dprefine \F$.	On the other hand, since $\id \not\le \e$ for any $\e\in \E$, we have $\F \not\leq_H \E$, which means
		$\E \not\le_P^s \F$. 	
	\end{proof}
	
	\subsection{Refinement Under Projector-Based Specifications}
	
	Finally, we examine refinement orders when specifications are further restricted to projectors, the most limited predicate class.
	
	\subsubsection{Predicate Transformers for Projectors}
	
	As in the deterministic case (Definition~\ref{def:wlp}), we can define weakest preconditions $wp^a$, weakest liberal preconditions $wlp^a$, and strongest postconditions $sp^a$ for nondeterministic programs when projectors are used as predicates. The following lemma provides explicit formulas.
	
	\begin{lemma}[Projector Transformers for Nondeterministic Programs]\label{lem:nwpwlpsp}
		For $\E \in \nprog(V)$ and projectors $P, Q \in \s(\h_W)$,
\[
wp^a.\E.Q = \bigwedge_{\e\in \E} E(\e^\dag(Q)),
\quad
wlp^a.\E.Q = \bigwedge_{\e\in \E} \mathcal{N}(\e^\dag(Q^\bot)),
\quad
sp^a.\E.P = \bigvee_{\e\in \E} \supp{\e(P)}
\]
	\end{lemma}
	\begin{proof}
		Note that for $*\in \{\tot, \pal\}$, \[
		\E\models_* (P,Q) \qquad \mbox{iff}\qquad \forall \e\in \E: \e\models_* (P,Q).
		\]
		Then the results follow directly from Lemma~\ref{lem:wpwlpsp} for individual programs.
	\end{proof}

	The predicate transformers for nondeterministic programs inherit properties from their deterministic counterparts.
	
	\begin{theorem}[Properties of Nondeterministic Projector Transformers]\label{thm:nptransform}
		For any $\E \in \nprog(V)$, projectors $P, Q \in \s(\h_W)$, and $xp^p \in \{wp^p, wlp^p\}$,
		\begin{enumerate}
			\item (Galois connection) $sp^p.\E.(wlp^p.\E.Q) \le Q$ and $P \le wlp^p.\E.(sp^p.\E.P)$;
			\item $wp^p.\E.Q = wlp^p.\E.Q \wedge wp^p.\E.I_W$;
			\item $wp^p.\E.0 = sp^p.\E.0 = 0$ and $wlp^p.\E.I_W = I_W$;
			\item (Monotonicity) If $P \le Q$, then $xp^p.\E.P \le xp^p.\E.Q$ and $sp^p.\E.P \le sp^p.\E.Q$;
			\item (Conjunction) $xp^p.\E.(P \wedge Q) = xp^p.\E.P \wedge xp^p.\E.Q$;
			\item  $sp^p.\E.(P \wedge Q) \le sp^p.\E.P \wedge sp^p.\E.Q$;
			\item $xp^p.\E.P \vee xp^p.\E.Q \le xp^p.\E.(P \vee Q)$;
			\item (Disjunction) $sp^p.\E.P \vee sp^p.\E.Q = sp^p.\E.(P \vee Q)$;
			\item (Frame axiom) If $R \in \s(\h_{W'})$ with $W' \cap (V \cup W) = \emptyset$,
			\begin{align*}
				wp^p.\E.(Q \otimes R) &= (wp^p.\E.Q) \otimes R,\\
				sp^p.\E.(P \otimes R) &= (sp^p.\E.P) \otimes R,\\
				wlp^p.\E.(Q \otimes R) &\ge (wlp^p.\E.Q) \otimes R.
			\end{align*}
		\end{enumerate}
	\end{theorem}
	
	\begin{proof}
		All properties follow from the corresponding properties for deterministic programs (Theorem~\ref{thm:ptransformdet}) combined with the meet/join characterizations in Lemma~\ref{lem:nwpwlpsp}. For instance, monotonicity of $wp^p.\E$ follows from monotonicity of each $wp^p.\e$ and the fact that the meet of monotonic functions is monotonic.
	\end{proof}
	
	\subsubsection{Refinement Characterization}
	
	We define refinement orders for projector-based specifications analogously to the effect case: $\E \atrefine \F$ if for all projectors $P, Q \in \s(\h_W)$,
	\[
	\E \atsat (\{P\}, \{Q\}) \quad \text{implies} \quad \F \atsat (\{P\}, \{Q\}),
	\]
	and similarly for $\E \aprefine \F$ with partial correctness.
	It is straightforward to verify that $\E \atrefine \F$ iff $\ \forall Q: wp^p.\E.Q \le wp^p.\F.Q$, and $\E \aprefine \F$ iff $\ \forall Q: wlp^p.\E.Q \le wlp^p.\F.Q$ (equivalently, $\forall P: sp^p.\F.P \le sp^p.\E.P$). 
	
	In the following, we extend Theorem~\ref{thm:detproj} to nondeterministic programs.
	To this end, we first extend the notion of termination space to nondeterministic programs.
		
		\begin{definition}[Termination Space for Nondeterministic Programs]
			For a nondeterministic quantum program $\E \in \nprog(V)$, the \emph{termination space} is defined as
			\[
			T_\E = \bigcap_{\e \in \E} T_\e = \bigcap_{\e \in \E} \left\{|\psi\rangle \in \h_V : \tr(\e(|\psi\rangle\langle\psi|)) = \tr(|\psi\rangle\langle\psi|)\right\}.
			\]
			This is the subspace of states on which \emph{all} possible behaviors of $\E$ preserve trace (i.e., terminate with certainty under demonic choice).
		\end{definition}
		
		Note that from Lemma~\ref{lem:nwpwlpsp}, we have
		\[
		T_\E = wp^p.\E.I_V = \bigwedge_{\e \in \E} wp^p.\e.I_V = \bigwedge_{\e \in \E} E(\e^\dagger(I_V)).
		\]
		Now we can state the complete characterization.
		
		\begin{theorem}[Projector-Based Refinement Characterization]\label{thm:ndetproj_complete}
			For any $\E, \F \in \nprog(V)$,
			\begin{enumerate}
				\item $\E \aprefine \F$ iff
				\[
				\bigvee_{\f \in \F} \spann(\f) \subseteq \bigvee_{\e \in \E} \spann(\e);
				\]
				
				\item $\E \atrefine \F$ iff $\ T_\E \le T_\F$ and
				\[
				\bigvee_{\f \in \F} \spann(\f \circ \p_{T_\E}) \subseteq \bigvee_{\e \in \E} \spann(\e \circ \p_{T_\E});
				\]
				
				\item If both $\E$ and $\F$ contain only trace-preserving super-operators, then $\E \atrefine \F$ iff $\ \E \aprefine \F$.
			\end{enumerate}
		\end{theorem}
		
		\begin{proof}
			(1): From Lemma~\ref{lem:nwpwlpsp}, $sp^p.\E.P = \bigvee_{\e \in \E} \supp{\e(P)}$. Thus
			\begin{align*}
				\E \aprefine \F
				&\quad\mbox{iff}\quad \forall P: sp^p.\F.P \le sp^p.\E.P\\
				&\quad\mbox{iff}\quad \forall P: \bigvee_{\f \in \F} \supp{\f(P)} \le \bigvee_{\e \in \E} \supp{\e(P)}.
			\end{align*}
			Following the argument in Theorem~\ref{thm:detproj}(1), this holds for all $P$ iff the inclusion of Kraus operator spans holds.
			
			Clauses (2) and (3) can be proved similarly by applying the same techniques employed in the proof of Theorem~\ref{thm:detproj}(2) and (3).
					\end{proof}

	As a result, weakening the state predicates from effects to projectors leads to a strictly weaker notion of refinement orders for nondeterministic quantum programs.

	\begin{proposition}[Projectors Weaker Than Effects]\label{prop:neimpp}
		For nondeterministic programs $\E, \F \in \nprog(V)$,
		\begin{enumerate}
			\item $\E \dtrefine \F$ implies $\E \atrefine \F$, but the converse is false;
			\item $\E \dprefine \F$ implies $\E \aprefine \F$, but the converse is false.
		\end{enumerate}
	\end{proposition}
	
	\begin{proof}
	The implications are straightforward.
	Furthermore, since a deterministic program $\e$ can be viewed as the singleton nondeterministic program $\{\e\}$, the counterexamples from Proposition~\ref{lem:eimpp} apply directly to show that the reverse parts of (1) and (2) are not true.
	\end{proof}

	\section{Conclusion}
	\label{sec:conclusion}
	
	In this paper, we have systematically investigated refinement orders for both deterministic and nondeterministic quantum programs. By analyzing the relationships between different refinement orders induced by various quantum predicate types (projectors, effects, and sets of effects) and correctness criteria (total and partial), we provide a clear and comprehensive map of the quantum refinement landscape. Our results establish precise characterizations of these orders in terms of intrinsic mathematical structures—complete positivity for deterministic programs, and Hoare and Smyth orders for nondeterministic programs, thereby offering a solid semantic foundation for the development of practical refinement calculi.
	
	On the practical side, our findings underscore the critical influence of predicate expressiveness on the granularity of refinement: for deterministic programs, effect-based and set-of-effect-based specifications yield the same orders, while projector-based specifications yield strictly weaker ones. For nondeterministic programs, the choice between effects, sets of effects, and projectors leads to a hierarchy of refinement strengths. These insights guide language designers and verification engineers in selecting appropriate predicate types to match their refinement goals. Moreover, the characterizations enable the transfer of well-established domain-theoretic results to quantum program refinement, facilitating the development of sound and complete proof rules.
	
	Several promising directions for future work emerge from this study:
\begin{enumerate}
	\item \emph{Extension to infinite-dimensional Hilbert spaces.} While our framework focuses on finite-dimensional spaces, quantum field theory and continuous-variable quantum computing require infinite-dimensional models. Extending our results to such settings presents a significant mathematical challenge but is essential for a fully general theory.
	\item \emph{Refinement for hybrid quantum–classical programs.} Hybrid algorithms that intertwine quantum and classical computation are increasingly important. Developing a unified refinement framework that handles both quantum and classical components, including their interplay, is a natural and valuable extension.
	\item \emph{Concrete refinement calculi and tool support.} Our semantic characterizations pave the way for concrete refinement logics and automated tool support. Designing practical refinement rules and integrating them into quantum programming environments would directly impact the rigorous development of quantum software.
\end{enumerate}

	\section*{Acknowledgement}
	
	This research was partially supported by the Innovation Program for Quantum Science and Technology under Grant No. 2024ZD0300500, the National Natural Science Foundation of China under Grant 92465202, and the Tianyan Quantum Computing Program. We also express our gratitude to Minbo Gao for his valuable insights on applying the Sion Minmax Theorem in the proof of Theorem~\ref{thm:equivalence}.

	\bibliography{ref}
	
\end{document}